\newif\ifRSPA
\RSPAfalse

\ifRSPA
\documentclass[]{rsproca_new}
\else
\documentclass[a4paper]{amsart}
\fi

\usepackage{mathtools,amssymb}
\usepackage{xparse}

\usepackage{hyperref}
\usepackage{enumitem}
\usepackage{amsthm}

\usepackage{bm,here,url}
\usepackage{graphicx}
\usepackage{tabularx}
\usepackage{booktabs}

\theoremstyle{plain}
	\newtheorem{thm}{Theorem}[section]
	\newtheorem{lem}[thm]{Lemma}
	\newtheorem{prop}[thm]{Proposition}

	\newtheorem{conj}[thm]{Conjecture}
	\newtheorem{ques}[thm]{Question}
 	\newtheorem{ex}[thm]{Example}
   	\newtheorem{as}[thm]{Assumption}
\theoremstyle{definition}
	
\theoremstyle{remark}

\newcommand{\bbR}{\mathbb{R}}
\newcommand{\bbZ}{\mathbb{Z}}

\newcommand{\calV}{\mathcal{V}}
\newcommand{\calE}{\mathcal{E}}
\newcommand{\calT}{\mathcal{T}}
\newcommand{\tp}{\mathsf{T}}

\DeclareMathOperator{\diag}{diag}
\DeclareMathOperator{\tr}{tr}
\DeclareMathOperator{\vol}{vol}

\ifRSPA
\titlehead{Research}
\fi

\begin{document}

\ifRSPA
\title{A mathematical approach to mechanical properties of networks in thermoplastic elastomers}
\else
\title[A mathematical approach to mechanical properties of networks]{A mathematical approach to mechanical properties of networks in thermoplastic elastomers}
\fi

\ifRSPA
\author{Ken'ichi Yoshida$^{1}$, Naoki Sakata$^2$, and Koya Shimokawa$^{2,3}$}
\else
\author[K.Yoshida, N.Sakata and K.Shimokawa]{Ken'ichi Yoshida$^{1}$, Naoki Sakata$^2$, and Koya Shimokawa$^{2,3}$}
\fi

\ifRSPA
\address{%
  $^1$International Institute for Sustainability with Knotted Chiral Meta Matter (WPI-SKCM$^2$), Hiroshima University, Hiroshima 739-8526, Japan\\
  $^2$Center for Soft Matter Physics, Ochanomizu University, Tokyo 112-8610, Japan\\
  $^3$Department of Mathematics, Ochanomizu University, Tokyo 112-8610, Japan
}

\subject{geometry, applied mathematics, materials science}

\else
\address{$^1$ International Institute for Sustainability with Knotted Chiral Meta Matter (WPI-SKCM$^2$), Hiroshima University, Hiroshima 739-8526, Japan}
\address{$^2$ Center for Soft Matter Physics, Ochanomizu University, Tokyo 112-8610, Japan}
\address{$^3$ Department of Mathematics, Ochanomizu University, Tokyo 112-8610, Japan}
\fi

\begin{abstract}
We employ a mathematical model to analyze stress chains in thermoplastic elastomers (TPEs) with a microphase-separated spherical structure composed of triblock copolymers. The model represents stress chains during uniaxial and biaxial extensions using networks of spherical domains connected by bridges. We advance previous research and discuss permanent strain and other aspects of the network. It explores the dependency of permanent strain on the extension direction, using the average of tension tensors to represent isotropic material behavior. 
The concept of deviation angle is introduced to measure network anisotropy and is shown to play an essential role in predicting permanent strain when a network is extended in a specific direction.
The paper also discusses methods to create a new network structure using various polymers.
\end{abstract}

\ifRSPA

\else
\maketitle
\fi

\ifRSPA
\begin{fmtext}
\fi

\section{Introduction}
\label{section:intro}

A mathematical model for network elastoplasticity is introduced in \cite{KY}.
The motivation for developing this model is to analyze the structure of the stress chains of thermoplastic elastomers (TPEs).
See \cite{Nakajima, Polymer} for recent research on the network structure of TPEs.
The microphase-separated structure of TPE, which is composed of triblock copolymers, has a spherical structure.

\ifRSPA
\end{fmtext}

\maketitle
\fi

The stress chain under uniaxial extension was modeled using a network whose nodes are spherical domains and edges are bridges connecting two domains.
Kodama and the first author~\cite{KY} described the behavior of the dynamic network resulting from the joining and splitting of spherical domains that occurs when a TPE is extended.
Furthermore, they also mathematically defined stresses, Young's moduli, and permanent strains, which quantify the mechanical properties.
As this model is consistent with the molecular dynamics (MD) simulation in \cite{Aoyagi,Polymer}, we will also use it for our study.
Kodama, Morita, and Kotani~\cite{KMK} recently computed and compared stress--strain curves for uniaxial extensions by the MD simulation and the mathematical model in \cite{KY}. 
The results are qualitatively consistent, and the computation by the mathematical model is significantly faster than the MD simulation. 
In addition, Hosoya, Ito, Nakajima, and Morita~\cite{HINM} conducted the MD simulation for TPEs obtained by blend of copolymers. 
This is comparable to the model by simple superposition of networks in \cite{KY}.

In this paper, we advance the previous research~\cite{KY} and discuss permanent strains, anisotropy, and other aspects of nets.
We first consider the permanent strains for networks that include only one vertex per period (Section~\ref{sec:single}).
The permanent strain depends on the direction of an extension. 
We show that the permanent strain can be negative if we extend a particular net in a certain direction.
Nevertheless, we see that the average of the permanent strain of a network for all directions is generally positive (Theorem~\ref{thm:average}).
The `mean permanent strain' for a network reflects the nature of isotropic material. 
From the discussion in Section~\ref{sec:single}, we define a `deviation angle' for general networks. The angles are expected to help measure the anisotropy of networks.
In particular, for networks that contain one vertex per period, the more the extension is in the direction to give a larger maximal deviation angle, the smaller its permanent strain (Proposition~\ref{prop:deviation}).
We next discuss how to 
construct the desired networks by blending several types of triblock copolymers (Section~\ref{sec:blend}).
In particular, we determine the sublattices in equilibrium (called `harmonic' networks) of the triangular and body-centered cubic lattices (Theorems~\ref{thm:subnet_2dim} and~\ref{thm:subnet_3dim}).
We then simulate the extensions of several $2$-dimensional nets and determine the permanent strains and anisotropies (Section~\ref{sec:simulation}).
In particular, we define a `local deviation angle' and discuss a relationship to the permanent strain in a particular extension direction.
Finally, we give a mathematical model for the stress chain of a biaxial extension and discuss the topological changes of the networks (Section~\ref{sec:biaxial}).

\section{Preliminaries}
\label{sec:prelim}

This section provides an introduction to the various concepts needed for this paper.
See also \cite{KY}.

\subsection{Network model for stress chain of TPE}

In this paper, we consider TPE consists of triblock copolymers.
See \cite{Fredrickson}.
In this section, we consider type ABA triblock copolymers.
TPE has a sphere structure, the shape of a hard domain of A monomers is a sphere and hard domains are connected by bridge chains of B monomers.
The {\em stress chain} of TPE is a union of hard domains and the bridges connecting them that are subject to high forces.
The network model for the stress chain of TPE is introduced in \cite{KY}.
A vertex of the network corresponds to a hard domain and an edge corresponds to a set of bridge chains between two hard domains.
See Figure \ref{fig:model}.

\begin{figure}[htb]
\begin{center}
\includegraphics[width=30em]{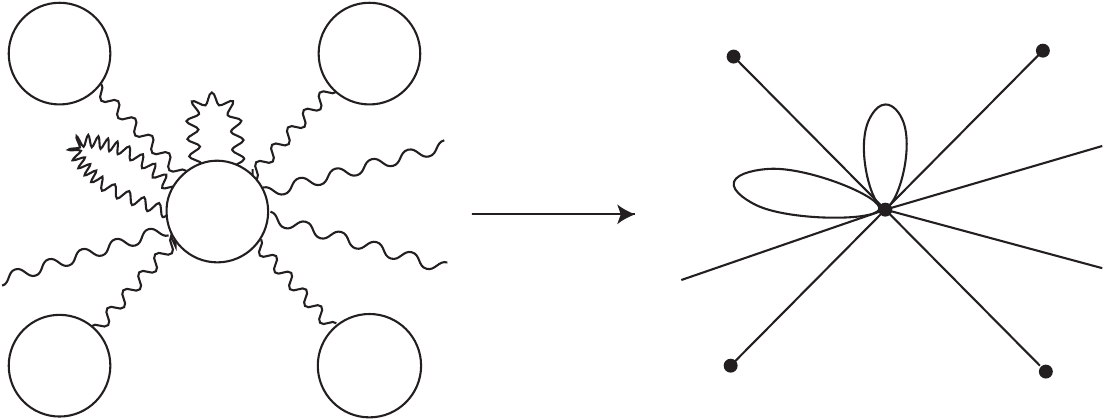}
\caption{The network model of the stress chain of TPE. Nodes of the network are hard domains and edges are bridges between two hard domains.}
\label{fig:model}
\end{center}
\end{figure}

One of the key ideas of the model is that a loop chain will become a bridging chain when the domain is split as shown in Figure \ref{fig:split}.
The splitting of a domain is observed in the simulation in \cite{Polymer}.
This type of network deformation cannot be observed in rubber and is a property unique to TPE.
We discuss how the properties of TPE are changed by this deformation.

\begin{figure}[htb]
\begin{center}
\includegraphics[width=30em]{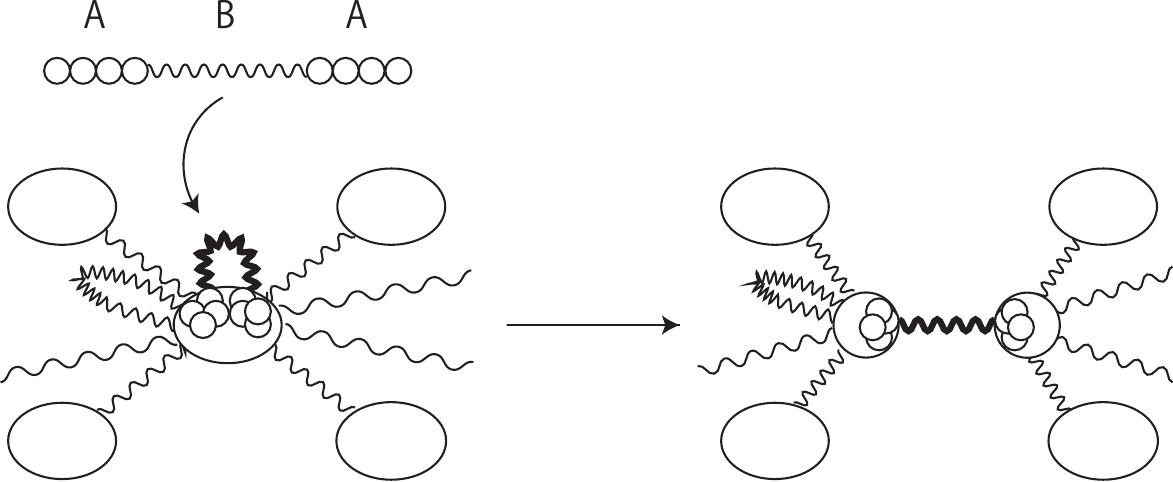}
\caption{Splitting of a hard domain. Part of the loops is converted to bridges and the rest remains loops.}
\label{fig:split}
\end{center}
\end{figure}

\subsection{Net}

The networks we are interested in are analyzed using the concept of nets.
Based on \cite{KS01, Sunada12, KY}, 
we prepare some mathematical notions for nets. 
Except for the notion of a subnet, 
the definitions are the same as in \cite{KY}. 
Let $X = (V,E,w)$ be an (abstract) weighted graph,
which is defined by the vertex set $V$ and the edge set $E$ 
with maps $o, t \colon E \to V$, and $\iota \colon E \to E$ 
such that $\iota^{2} = id$, $\iota (e) \neq e$, and 
$o(\iota (e)) = t(e)$ for any $e \in E$. 
The maps $o$ and $t$ associate the origin and the terminal of an edge, respectively. 
The map $\iota$ reverses the orientation of an edge. 
We allow a loop (edge $e$ such that $o(e) = t(e)$) 
and a multi-edge (edges with common terminal points). 
The weight function $w \colon E \to \bbR_{\geq 0}$ 
satisfies $w (\iota (e)) = w(e)$ for $e \in E$. 
We regard the weight of an edge as the number of edges. 
Hence we may replace an edge $e_{0}$ with the union of edges $e_{1}$ and $e_{2}$ 
if $o(e_{0}) = o(e_{1}) = o(e_{2}), t(e_{0}) = t(e_{1}) = t(e_{2})$, and $w(e_{0}) = w(e_{1}) + w(e_{2})$. 
The weight function is often omitted in the notation. 
The degree of a vertex $v \in V$ is defined by $\deg (v) = \sum_{o(e)=v} w(e)$. 
Note that the weight of a loop contributes twice to the degree of its endpoint. 

We can naturally identify $X$ with a one-dimensional complex. 
Note that two elements $e$ and $\iota (e)$ in $E$ correspond to 
a single one-cell in the complex. 
We may reduce the complex by removing the zero-weight edges. 
If this reduced complex is connected, we say that the graph is connected. 

We consider a connected graph $X$.
For $N \geq 1$, suppose that $L = \bbZ^{N}$ acts on $X$ 
as (weight-preserving) automorphisms of the graph, 
the quotient map $\omega \colon X \to X/L = (V/L, E/L)$ is a covering, 
and $X/L$ is a graph. 
Then we say that $X$ is a \emph{periodic graph}. 
A map $\Phi \colon V \to \bbR^{N}$ 
is called a \emph{periodic realization} of $X$ in $\bbR^{N}$ 
if there exists an injective homomorphism 
$\rho \colon L \hookrightarrow \mathbb{R}^{N}$ as $\bbZ$-modules 
satisfying that 
\begin{enumerate}[label=(\roman*)]
	\item $\Phi$ is $L$-equivariant, 
	i.e. $\Phi (\gamma v) = \Phi (v) + \rho (\gamma)$ for any $v \in V$ and $\gamma \in L$, and 
	\item $\rho (L)$ is discrete in $\bbR^{N}$. 
\end{enumerate}
We call $\rho$ the \emph{period homomorphism} for $\Phi$. 
The pair $(X, \Phi)$ is called a \emph{net} in $\bbR^{N}$ 
if $\Phi$ is a periodic realization of a periodic graph $X$ in $\bbR^{N}$. 
A periodic realization $\Phi$ maps an edge $e \in E$ to a vector 
\[
\Phi (e) = \Phi (t(e)) - \Phi (o(e)) \in \bbR^{N}. 
\] 
The equivariancy induces the map $\Phi \colon E/L \to \bbR^{N}$ 
by $\Phi (\omega (e)) = \Phi (e)$. 
If $e$ is a loop, then $\Phi (e) = 0$. 

A net $(X^{\prime}, \Phi^{\prime}, \rho^{\prime})$ is called 
a \emph{subnet} of a net $(X, \Phi, \rho)$ if 
\begin{enumerate}[label=(\roman*)]
    \item $X^{\prime} = (V^{\prime}, E^{\prime}, w^{\prime})$ 
    is a subgraph of $X = (V,E,w)$, 
    i.e. $V^{\prime} \subset V$, $E^{\prime} \subset E$, 
    and $w^{\prime}(e) \leq w(e)$ for any $e \in E^{\prime}$, 
    \item $\Phi^{\prime}$ is the restriction of $\Phi$ to $V^{\prime}$, and 
    \item $\rho^{\prime}(L)$ is a subgroup of $\rho (L)$. 
\end{enumerate}

The \emph{energy} (per period) of a net $(X, \Phi)$ is defined by 
\[
\calE(X, \Phi) = 
\frac{1}{2}\sum_{e \in E/L} w(e) \| \Phi (e) \|^{2}. 
\]
A periodic realization $\Phi$ of $X$ is \emph{harmonic} 
if the energy $\mathcal{E}(X, \Phi)$ is minimal among the periodic realizations of $X$ 
with the fixed period homomorphism $\rho$. 
A periodic realization $\Phi$ of $X$ is \emph{standard} 
if the energy $\calE(X, \Phi)$ is minimal among the periodic realizations of $X$ 
with the fixed covolume $\vol(\bbR^{N} / \rho (L))$.

For a net $(X,\Phi)$, the local and global tension tensors were introduced in \cite{KY}. 
For a vertex $v$ of $X$ or $X/L$, 
the \emph{local tension tensor} is defined by 
\[
\calT(v) = \sum_{o(e)=v} w(e) \Phi (e)^{\otimes 2}, 
\]
where 
\[
\begin{pmatrix}
x_{1} \\ \vdots \\ x_{N}
\end{pmatrix}^{\otimes 2}
=
\begin{pmatrix}
x_{1} \\ \vdots \\ x_{N}
\end{pmatrix}
\otimes 
\begin{pmatrix}
x_{1} \\ \vdots \\ x_{N}
\end{pmatrix}
= 
\begin{pmatrix}
x_{1} \\ \vdots \\ x_{N}
\end{pmatrix}
(x_{1}, \dots , x_{N}) = 
\begin{pmatrix}
x_{1}^{2} & \cdots & x_{1}x_{N} \\
\vdots & \ddots & \vdots \\
x_{N}x_{1} & \cdots & x_{N}^{2}
\end{pmatrix}.
\] 
The \emph{global tension tensor} (per period) is defined by 
\[
\calT(X, \Phi) = \frac{1}{2} \sum_{v \in V/L} \calT(v). 
\]
It holds that 
$\tr(\calT(X, \Phi)) = \calE(X, \Phi)$.

A standard realization is characterized as follows 
\cite[Theorem 3.3]{KY} 
(cf. \cite[Theorem 7.5]{Sunada12}): 
A periodic realization $\Phi$ is standard if and only if it is harmonic 
and the global tension tensor $\calT(X, \Phi)$ is a constant multiple of the identity matrix. 
Consequently, a standard realization is unique up to similar transformations.

\subsection{Stress}

For a net $(X,\Phi)$, 
let $\calT = \calT(X,\Phi)$ 
and $\calV = \vol(\bbR^{N}/\rho(L))$ 
denote the (global) tension tensor and the volume per period. 
In \cite[Section 4]{KY}, 
the tension tensor $\calT$ was interpreted as a stress tensor by a physical argument. 
We regard the energy $\calE = \calE(X,\Phi)$ as physical energy. 
We can imagine that an edge of a net is a spring with zero rest length. 
More precisely, an edge of a net corresponds to a polymer chain. 
Its random motion induces an entropic force, which is determined by the distance between the endpoints. 
The energy $\calE$ corresponds to the Helmholtz free energy for entropic elasticity. 
Moreover, the energy per volume $\calE / \calV$ induces the strain energy density function for a neo-Hookean material, 
which is one of the simplest hyperelastic materials. 
For more details of hyperelastic model, 
we refer the reader to \cite{Ogden, Treloar}. 
Let us consider a macro-scale object consisting of the net $(X,\Phi)$ at the micro-scale. 
Then the Cauchy stress tensor of this object is given by 
$(2 /\calV) \calT$. 
In this setting, a harmonic net is an equilibrium state. 
We assume that the volume $\calV$ is preserved. 
It is well known that the volume of an ordinary TPE changes little during deformation, similar to an ordinary rubber. 
The MD simulation in \cite{Aoyagi,Polymer} also assumes that the volume is preserved. 
Under this assumption, this object satisfies the energy density function for an incompressible neo-Hookean material. 
Then a standard net is an equilibrium state without external force.

The stress under a uniaxial extension is given as follows. 
Let $\calT = \calT(X,\Phi) = (\tau_{ij})_{1 \leq i,j \leq N}$. 
Note that $\Phi$ is not necessarily standard. 
For $\lambda >0$, the diagonal matrix 
\[
A(\lambda) = \diag \left( \lambda, \lambda^{-1/(N-1)}, \dots, \lambda^{-1/(N-1)} \right) 
\in SL(N, \mathbb{R}) 
\]
induces a uniaxial extension with strain $\epsilon = \lambda -1$. 
The matrix $A(\lambda)$ acts the net $(X,\Phi)$ 
by $A(\lambda) (X,\Phi) = (X, A(\lambda) \circ \Phi)$. 
Then the \emph{engineering stress} with strain $\epsilon = \lambda -1$ is given by 
\[
\sigma_{\mathrm{eng}} = 
\dfrac{2}{\calV} \left( \tau_{11}\lambda - \dfrac{1}{N-1} (\tau_{22} + \dots + \tau_{NN})\lambda^{-1-2/(N-1)} \right). 
\] 
The \emph{permanent strain} is the number $\epsilon_{0}$ satisfying 
$\sigma_{\mathrm{eng}}(1+\epsilon_{0}) = 0$. 
Hence 
\[
\epsilon_{0} = \left( \frac{\tau_{22} + \dots + \tau_{NN}}{(N-1)\tau_{11}} \right)^{(N-1)/2N} -1
\]
as shown in \cite[Proposition 4.2]{KY}. 
We will consider a situation that a net is deformed to another net under an extension. 
In such a case, the resulting net usually has a positive permanent strain. 
Thus the permanent strain represents a kind of plasticity. 
In \cite[Proposition 8.2]{KY}, however, 
it was noted that the permanent strain may be negative. 
We will discuss this in the next section.

\section{Permanent strain after a single splitting}
\label{sec:single}
In this section, 
we consider a net with a single vertex per period. 
We discuss the permanent strain after a single splitting of the vertex. 

Let $X_{0}$ be a periodic graph such that only a single vertex exists in each period. 
We identify each $\bm{i} = (i_{1}, \dots , i_{N}) \in \mathbb{Z}^{N}$ with a vertex of $X_{0}$. 
Let $e_{\bm{i}}$ denote the edge of $X_{0}$ joining $0$ and $\bm{i} \in \bbZ^{N}$. 
We write $w_{\bm{i}}$ for the weight of $e_{\bm{i}}$. 
It is necessary that $w_{\bm{i}} = w_{-\bm{i}}$. 
Let $\rho \colon \bbZ^{N} \to \bbR^{N}$ be a period homomorphism, 
and let $(\bm{u}_{1}, \dots , \bm{u}_{N})$ be a basis of $\rho (\bbZ^{N})$. 
The period homomorphism $\rho$ induces a periodic realization $\Phi_{0}$ of $X_{0}$ 
such that the image of vertices is $\rho (\bbZ^{N})$. 
For $\bm{i} \in \bbZ^{N}$, 
the edge $e_{\bm{i}}$ corresponds to 
$\bm{v}_{\bm{i}} = i_{1}\bm{u}_{1} + \dots + i_{N}\bm{u}_{N} \in \bbR^{N}$. 
Then $\Phi_{0}$ is a harmonic realization. 

Let $I \subset \bbZ^{N}$ such that $\bbZ^{N} = I \sqcup -I \sqcup \{ 0 \}$. 
Let $X_{1}$ be a periodic graph obtained from $X_{0}$ 
by splitting the vertex $0$ into $v_{0}$ and $v_{1}$ 
so that $v_{0}$ and $v_{1}$ are endpoints of $e^{\prime}_{\bm{i}}$ 
for $\bm{i} \in -I$ and $\bm{i} \in I$, respectively,
where $e^{\prime}_{\bm{i}}$ is an edge of $X_{1}$ obtained from $e_{\bm{i}}$. 
We write $w_{0}^{\prime}$ for the weight of the new non-loop edge. 
For $\bm{u} \neq 0 \in \bbR^{N}$, 
let $\bbR^{N}_{\bm{u}} = \{ \bm{v} \in \bbR^{N} \mid \bm{u} \cdot \bm{v} > 0 \}$ 
and $I_{\bm{u}} = \bbR^{N}_{\bm{u}} \cap \bbZ^{N}$. 
If $I_{\bm{u}} \subset I$, 
the vertices split in the direction of $\bm{u}$. 
Let $\Phi_{1}$ be a harmonic realization of $X_{1}$ with the period homomorphism $\rho$. 
We may assume that $\Phi_{1}(v_{0}) = 0$. Let $\bm{x} = \Phi_{1}(v_{1})$. 

Theorem 7.2 and Proposition 8.3 in \cite{KY} implies that 
\begin{align*}
\bm{x} 
& = \frac{\sum_{\bm{i} \in I}w_{\bm{i}}\bm{v}_{\bm{i}}}{w_{0}^{\prime} + \sum_{\bm{i} \in I}w_{\bm{i}}}, \\
\calT (X_{0},\Phi_{0}) - \calT (X_{1},\Phi_{1}) 
& = \frac{\left(\sum_{\bm{i} \in I}w_{\bm{i}}\bm{v}_{\bm{i}}\right)^{\otimes 2}}{w_{0}^{\prime} + \sum_{\bm{i} \in I}w_{\bm{i}}}
= \left( w_{0}^{\prime} + \sum_{\bm{i} \in I}w_{\bm{i}} \right) \bm{x}^{\otimes 2}. 
\end{align*}
Moreover, $\bm{z} = \sum_{\bm{i} \in I}w_{\bm{i}}\bm{v}_{\bm{i}}$ 
does not depend on $w_{0}^{\prime}$. 

Suppose that $(X_{0},\Phi_{0})$ is a standard net. 
Let $\calT (X_{1},\Phi_{1}) = (\tau_{ij})_{1 \leq i,j \leq N}$. 
Consider a uniaxial extension in the direction of the first coordinate. 
In this case, the vertices split in the direction of $\bm{u} =(1,0, \dots, 0)^{\tp}$. 
Then the permanent strain is 
 \[
\epsilon_{0} = \left( \frac{\tau_{22} + \dots + \tau_{NN}}{(N-1)\tau_{11}} \right)^{(N-1)/2N} -1. 
\]
For general $\bm{u}$, the permanent strain is obtained by rotation.

The permanent strain may be negative. 
This results from the difference between the directions of $\bm{u}$ and $\bm{z}$. 

\begin{prop}
\label{prop:deviation}
The permanent strain $\epsilon_{0}$ is negative 
if and only if the angle between $\bm{u}$ and $\bm{z}$ is more than $\arccos (1/\sqrt{N})$. 
\end{prop}
\begin{proof}
Suppose that $\bm{u} = (1,0, \dots, 0)^{\tp}$. 
Let $\bm{z} = (z_{1}, \dots, z_{N})^{\tp}$. 
Since $(X_{0},\Phi_{0})$ is a standard net, 
\begin{align*}
\epsilon_{0} < 0 & \Longleftrightarrow (N-1)\tau_{11} > \tau_{22} + \dots + \tau_{NN} \\
& \Longleftrightarrow (N-1) z_{1}^{2} < z_{2}^{2} + \dots + z_{N}^{2} \\
& \Longleftrightarrow \frac{z_{1}}{\| \bm{z} \|} < \frac{1}{\sqrt{N}} \\
& \Longleftrightarrow \bm{u} \cdot \frac{\bm{z}}{\| \bm{z} \|} < \frac{1}{\sqrt{N}}. 
\end{align*}
\end{proof}

\begin{ex} 
Suppose that $w_{\bm{i}} = 0$ 
for $\bm{i} \in \bbZ^{2} \setminus \{(0,0), (\pm 1, 0), (0, \pm 1)\}$, 
$\bm{u}_{1} = (l_{1}, 0)$, and $\bm{u}_{2} = (0, l_{2})$. 
Let $w_{1} = w_{(1,0)}, w_{2} = w_{(0,1)}$. 
Consider the permanent strain $\epsilon_{0}$ 
after the splitting in the direction $\bm{u} = (x,y)^{\tp}$. 
Suppose that $w_{1} < w_{2}$ and $\theta = \arctan (y/x) \in (0, \pi/2)$. 
Then $\epsilon_{0} < 0$ if and only if 
\[
\theta < \dfrac{1}{2} \arctan \dfrac{w_{2} - w_{1}}{2\sqrt{w_{1}w_{2}}}. 
\]
\end{ex}
\begin{proof}
Since the net $(X_{0}, \Phi_{0})$ is standard and 
$\calT(X_{0}, \Phi_{0}) = 
\begin{pmatrix}
w_{1}l_{1}^{2} & 0 \\
0 & w_{2}l_{2}^{2}
\end{pmatrix}$, 
we have $w_{1}l_{1}^{2} = w_{2}l_{2}^{2}$. 
Hence we can write $l_{1} = w_{1}^{-1/2}l$ and $l_{2} = w_{2}^{-1/2}l$. 
Then $\bm{z} = (w_{1}l_{1}, w_{2}l_{2}) = (w_{1}^{1/2}l, w_{2}^{1/2}l)$. 

We may assume that $x = \cos \theta$ and $y = \sin \theta$. 
By Proposition~\ref{prop:deviation}, 
it holds that $\epsilon_{0} < 0$ 
if and only if the angle between $\bm{u}$ and $\bm{z}$ is more than $\pi/4$. 
Hence we consider the condition 
$\dfrac{\bm{u} \cdot \bm{z}}{\|\bm{u}\| \|\bm{z}\|} < \dfrac{1}{\sqrt{2}}$. 
In this case, we have $\theta < \pi/4$ since $w_{1} < w_{2}$. 
For $0 < \theta < \pi/4$, 
we have 
\begin{align*}
\dfrac{\bm{u} \cdot \bm{z}}{\|\bm{u}\| \|\bm{z}\|} < \dfrac{1}{\sqrt{2}} 
& \Longleftrightarrow 
\sqrt{w_{1}} \cos \theta + \sqrt{w_{2}} \sin \theta < \dfrac{1}{\sqrt{2}} \sqrt{w_{1} + w_{2}} \\
& \Longleftrightarrow 
w_{1} + 2\sqrt{w_{1}w_{2}} \tan \theta + w_{2} \tan^{2} \theta 
< \dfrac{w_{1} + w_{2}}{2} (1 + \tan^{2} \theta) \\
& \Longleftrightarrow 
\tan 2\theta = \dfrac{2\tan \theta}{1- \tan^{2} \theta} < \dfrac{w_{2} - w_{1}}{2\sqrt{w_{1}w_{2}}} \\
& \Longleftrightarrow 
\theta < \dfrac{1}{2} \arctan \dfrac{w_{2} - w_{1}}{2\sqrt{w_{1}w_{2}}}. 
\end{align*}
The last inequality implies that $\theta < \pi/4$. 
\end{proof}

Fix a net $(X_{0}, \Phi_{0})$. 
The new net $(X_{1}, \Phi_{1})$ depends on the direction $\bm{u} \in S^{N-1}$ of extension. 
To consider how often the permanent strain is negative, 
let us take the average of the tension tensors. 
More precisely, 
for $Q \in SO(N)$, 
suppose that the net $(X_{0}, Q^{-1} \circ \Phi_{0})$ induces the new net $(X_{1}(Q), \Phi_{1}(Q))$ 
by the splitting in the direction $\bm{u} = (1,0, \dots, 0)^{\tp}$. 
Then the average of the tension tensors is given by 
\[
\overline{\calT(X_{1}, \Phi_{1})} = (\overline{\tau_{ij}})_{1 \leq i,j \leq N} 
= \int_{Q \in SO(N)} \calT(X_{1}(Q), \Phi_{1}(Q)) d\mu, 
\]
where $\mu$ is the unit Haar measure on $SO(N)$. 
Then the permanent strain for this average is 
\[
\overline{\epsilon_{0}} = \left( \frac{\overline{\tau_{22}} + \dots + \overline{\tau_{NN}}}{(N-1)\overline{\tau_{11}}} \right)^{(N-1)/2N} -1. 
\]
We call this {\em the mean permanent strain}.
The following theorem implies that 
the permanent strain is not so often negative.

\begin{thm}
\label{thm:average}
Let $N = 2$ or $3$. 
Let $(X_{0}, \Phi_{0})$ be a net such that only a single vertex exists in each period. 
Suppose that $w_{\bm{i}} > 0$ for only finitely many $\bm{i} \in \bbZ^{N}$. 
Then the mean permanent strain is positive, i.e. $\overline{\epsilon_{0}} > 0$. 
\end{thm}

\begin{proof}
It is sufficient to show that 
$(N-1)\overline{\tau_{11}} < \overline{\tau_{22}} + \dots + \overline{\tau_{NN}}$. 
Let $\{ \pm \bm{v}_{1}, \dots, \pm \bm{v}_{n} \}$ be the set of $\bm{v}_{\bm{i}}$'s 
for $\bm{i} \in \bbZ^{N} \setminus 0$ with non-zero $w_{\bm{i}}$. 
Let $w_{k} > 0$ denote the weight for $\bm{v}_{k}$. 
For $Q \in SO(N)$,  the splitting after the rotation by $Q^{-1}$ induces 
$\bm{z}(Q) = Q^{-1} \sum_{k=1}^{n} \epsilon_{k}(Q) w_{k} \bm{v}_{k}$, 
where $\epsilon_{k}(Q) = \pm 1, 0$ is the signature of $(Q \bm{u}) \cdot \bm{v}_{k}$. 
Note that $\epsilon_{k}(Q) = \pm 1$ for generic $Q$. 
Then 
\[
\calT(X_{1}(Q), \Phi_{1}(Q)) = \calT(X_{0}, \Phi_{0}) - \frac{\bm{z}(Q)^{\otimes 2}}{w_{0}^{\prime} + \sum_{k=1}^{n}w_{k}}. 
\]
The average is given by 
\begin{align*}
\overline{\calT(X_{1}, \Phi_{1})}  & = \calT(X_{0}, \Phi_{0}) - \frac{\overline{\bm{z}^{\otimes 2}}}{w_{0}^{\prime} + \sum_{k=1}^{n}w_{k}}, \\ 
\text{where} \quad 
\overline{\bm{z}^{\otimes 2}} & = \int_{Q \in SO(N)} \bm{z}(Q)^{\otimes 2} d\mu. 
\end{align*}
Let $\overline{\bm{z}^{\otimes 2}} = (z_{ij})_{1 \leq i,j \leq 2}$ 
and $\delta = (N-1) z_{11} - z_{22} - \dots - z_{NN}$. 
Since $(X_{0}, \Phi_{0})$ is standard, it is sufficient to show that $\delta > 0$.

We have 
\begin{align*}
\overline{\bm{z}^{\otimes 2}} & = 
\int_{Q \in SO(N)} 
Q^{-1} 
(\sum_{k=1}^{n} \epsilon_{k}(Q) w_{k} \bm{v}_{k})^{\otimes 2} 
Q d\mu, \\
(\sum_{k=1}^{n} \epsilon_{k}(Q) w_{k} \bm{v}_{k})^{\otimes 2} 
& = \sum_{k=1}^{n} (w_{k} \bm{v}_{k})^{\otimes 2} 
+ \sum_{k \neq l} (\epsilon_{k}(Q) w_{k} \bm{v}_{k}) 
\otimes (\epsilon_{l}(Q) w_{l} \bm{v}_{l}) \\
& = -(n-2) \sum_{k=1}^{n} (w_{k} \bm{v}_{k})^{\otimes 2} 
+ \sum_{k < l} (\epsilon_{k}(Q) w_{k} \bm{v}_{k} 
+ \epsilon_{l}(Q) w_{l} \bm{v}_{l})^{\otimes 2}. 
\end{align*}
Consider the contribution to $\delta$ from each term of the last expression. 
The contribution from $(w_{k} \bm{v}_{k})^{\otimes 2}$ 
is zero by Lemma~\ref{lem:avn1}. 
The contribution from 
$(\epsilon_{k}(Q) w_{k} \bm{v}_{k} + \epsilon_{l}(Q) w_{l} \bm{v}_{l})^{\otimes 2}$ 
is positive by Lemma~\ref{lem:avn2}. 
Hence $\delta > 0$. 
\end{proof}

In the next two lemmas, 
we do not assume that $(X_{0}, \Phi_{0})$ is standard. 
Nonetheless, the number $\delta$ is defined in the same manner. 

\begin{lem}
\label{lem:avn1}
Suppose that $n = 1$, i.e. 
$\{ \pm \bm{v}_{1} \}$ is the set of $\bm{v}_{\bm{i}}$'s 
for $\bm{i} \in \bbZ^{N} \setminus 0$ with non-zero $w_{\bm{i}}$. 
Then $\delta = 0$. 
\end{lem}

\begin{lem}
\label{lem:avn2}
Suppose that $n = 2$, i.e. 
$\{ \pm \bm{v}_{1}, \pm \bm{v}_{2} \}$ is the set of $\bm{v}_{\bm{i}}$'s 
for $\bm{i} \in \bbZ^{N} \setminus 0$ with non-zero $w_{\bm{i}}$. 
Suppose that the vectors $\bm{v}_{1}$ and $\bm{v}_{2}$ are linearly independent. 
Then $\delta > 0$. 
\end{lem}

In the case $N=2$, 
let 
$R(\theta) = 
\begin{pmatrix}
\cos \theta & -\sin \theta \\
\sin \theta & \cos \theta 
\end{pmatrix}
\in SO(2)$ 
for $0 \leq \theta \leq 2\pi$. 
The splitting after the rotation induces 
$\bm{z}(\theta) = 
R(\theta)^{-1}
\sum_{k=1}^{n} \epsilon_{k}(\theta, \varphi, \psi) w_{k} \bm{v}_{k}$. 
Then 
\[
\overline{\bm{z}^{\otimes 2}} = \frac{1}{\pi} \int_{0}^{\pi} \bm{z}(\theta)^{\otimes 2} d\theta. 
\]

In the case $N=3$, 
let 
\[
R(\theta, \varphi, \psi) = 
\begin{pmatrix}
\cos \varphi & -\sin \varphi & 0 \\
\sin \varphi & \cos \varphi & 0 \\
0 & 0 & 1 
\end{pmatrix}
\begin{pmatrix}
1 & 0 & 0 \\
0 & \cos \theta & -\sin \theta \\
0 & \sin \theta & \cos \theta 
\end{pmatrix}
\begin{pmatrix}
\cos \psi & -\sin \psi & 0 \\
\sin \psi & \cos \psi & 0 \\
0 & 0 & 1 
\end{pmatrix}
\]
for $0 \leq \theta \leq \pi, 0 \leq \varphi \leq 2\pi, 0 \leq \psi \leq 2\pi$. 
The splitting after the rotation induces 
$\bm{z}(\theta, \varphi, \psi) = 
R(\theta, \varphi, \psi)^{-1}
\sum_{k=1}^{n} \epsilon_{k}(\theta, \varphi, \psi) w_{k} \bm{v}_{k}$. 
Then 
\[
\overline{\bm{z}^{\otimes 2}} = 
\frac{1}{8\pi^{2}} \int_{0}^{\pi} \int_{0}^{2\pi} \int_{0}^{2\pi} \bm{z}(\theta, \varphi, \psi)^{\otimes 2} 
\sin \theta d\theta d\varphi d\psi. 
\]

\begin{proof}[Proof of Lemma~\ref{lem:avn1}]
In the case $N=2$, 
we may assume that $w_{1} \bm{v}_{1} = (1,0)^{\tp}$. 
Then 
\begin{align*}
\overline{\bm{z}^{\otimes 2}} & = 
\frac{1}{2\pi} \int_{0}^{2\pi} 
R(\theta)^{-1}
\begin{pmatrix}
1 & 0 \\
0 & 0 
\end{pmatrix}
R(\theta)
d\theta \\ 
& = 
\frac{1}{2}
\begin{pmatrix}
1 & 0 \\
0 & 1 
\end{pmatrix}. 
\end{align*}
Hence $\delta = 0$. 

In the case $N=3$, 
we may assume that $w_{1} \bm{v}_{1} = (1,0,0)^{\tp}$. 
Then 
\begin{align*}
\overline{\bm{z}^{\otimes 2}} & = 
\frac{1}{8\pi^{2}} \int_{0}^{\pi} \int_{0}^{2\pi} \int_{0}^{2\pi} 
R(\theta, \varphi, \psi)^{-1}
\begin{pmatrix}
1 & 0 & 0 \\
0 & 0 & 0 \\
0 & 0 & 0 
\end{pmatrix}
R(\theta, \varphi, \psi)
\sin \theta d\theta d\varphi d\psi \\ 
& = 
\frac{1}{3}
\begin{pmatrix}
1 & 0 & 0 \\
0 & 1 & 0 \\
0 & 0 & 1 
\end{pmatrix}. 
\end{align*}
Hence $\delta = 0$. 
\end{proof}

\begin{proof}[Proof of Lemma~\ref{lem:avn2}]

In the case $N=2$, 
we may assume that 
$w_{1} \bm{v}_{1} = (1,0)^{\tp}$, $w_{2} \bm{v}_{2} = (x,y)^{\tp}$, 
$x \geq 0$, and $y > 0$. 
Let $\theta_{0} = \arctan (y / x) \in (0, \pi/2)$. 
Then 
\begin{align*}
\overline{\bm{z}^{\otimes 2}} & = 
\frac{1}{2\pi} \left( 
\int_{\theta_{0} - \pi/2}^{\pi/2} 
R(\theta)^{-1}
(w_{1}\bm{v}_{1} + w_{2}\bm{v}_{2})^{\otimes 2}
R(\theta)
d\theta \right. \\
& \quad + 
\int_{\pi/2}^{\theta_{0} + \pi/2} 
R(\theta)^{-1}
(- w_{1}\bm{v}_{1} + w_{2}\bm{v}_{2})^{\otimes 2}
R(\theta)
d\theta \\
& \quad + 
\int_{\theta_{0} + \pi/2}^{3\pi/2} 
R(\theta)^{-1}
(- w_{1}\bm{v}_{1} - w_{2}\bm{v}_{2})^{\otimes 2}
R(\theta)
d\theta \\
& \left. \quad + 
\int_{3\pi/2}^{\theta_{0} + 3\pi/2} 
R(\theta)^{-1}
(w_{1}\bm{v}_{1} - w_{2}\bm{v}_{2})^{\otimes 2}
R(\theta)
d\theta \right) \\ 
& = 
\frac{1}{\pi} \left( 
\int_{\theta_{0} - \pi/2}^{\pi/2} 
R(\theta)^{-1}
\begin{pmatrix}
1 + x \\
y 
\end{pmatrix}^{\otimes 2}
R(\theta)
d\theta \right. \\
& \left. \quad + 
\int_{\pi/2}^{\theta_{0} + \pi/2} 
R(\theta)^{-1}
\begin{pmatrix}
-1 + x \\
y 
\end{pmatrix}^{\otimes 2}
R(\theta)
d\theta \right). 
\end{align*}

Let $\alpha_{ij} \in \bbR$ for $1 \leq i, j \leq 2$. 
Suppose that $\alpha_{ij} = \alpha_{ji}$. 
Let 
\[
\begin{pmatrix}
\beta_{11} & \beta_{12} \\
\beta_{21} & \beta_{22} 
\end{pmatrix} 
= \int R(\theta)^{-1}
\begin{pmatrix}
\alpha_{11} & \alpha_{12} \\
\alpha_{21} & \alpha_{22} 
\end{pmatrix} 
R(\theta) d\theta. 
\]
Then 
\[
\beta_{11} - \beta_{22} = 
\frac{1}{2} (\alpha_{11} - \alpha_{22}) \sin 2\theta - \alpha_{12} \cos 2\theta. 
\]
Since 
\[
\sin 2\theta_{0} = \frac{2xy}{x^{2} + y^{2}} \quad \text{and} \quad 
\cos 2\theta_{0} = \frac{x^{2} - y^{2}}{x^{2} + y^{2}}, 
\]
we have 
\begin{align*}
\delta & = 
\frac{1}{\pi} \left(
\left[
\frac{1}{2} ((1+x)^{2} - y^{2}) \sin 2\theta - (1+x)y \cos 2\theta 
\right]_{\theta_{0} - \pi/2}^{\pi/2} \right. \\
& \left. \quad + 
\left[
\frac{1}{2} ((-1+x)^{2} - y^{2}) \sin 2\theta - (-1+x)y \cos 2\theta 
\right]_{\pi/2}^{\theta_{0} + \pi/2} \right) \\
& = \frac{1}{\pi} (2x \sin 2\theta_{0} + 2y (1 - \cos 2\theta_{0})) \\
& = \frac{4y}{\pi}  > 0. 
\end{align*}

In the case $N=3$, 
we may assume that 
$w_{1} \bm{v}_{1} = (1,0,0)^{\tp}$, $w_{2} \bm{v}_{2} = (x,y,0)^{\tp}$, 
$x \geq 0$, and $y > 0$. 
Let $\varphi_{0} = \arctan (y / x) \in (0, \pi/2)$. 
Then 
\begin{align*}
\overline{\bm{z}^{\otimes 2}} 
& = \frac{1}{8\pi^{2}} \left( 
\int_{0}^{\pi} \int_{\varphi_{0} - \pi/2}^{\pi/2} \int_{0}^{2\pi}
R(\theta, \varphi, \psi)^{-1}
(w_{1}\bm{v}_{1} + w_{2}\bm{v}_{2})^{\otimes 2}
R(\theta, \varphi, \psi)
\sin \theta d\theta d\varphi d\psi \right. \\
& \quad + 
\int_{0}^{\pi} \int_{\pi/2}^{\varphi_{0} + \pi/2}  \int_{0}^{2\pi}
R(\theta, \varphi, \psi)^{-1}
(- w_{1}\bm{v}_{1} + w_{2}\bm{v}_{2})^{\otimes 2}
R(\theta, \varphi, \psi)
\sin \theta d\theta d\varphi d\psi \\
& \quad + 
\int_{0}^{\pi} \int_{\varphi_{0} + \pi/2}^{3\pi/2} \int_{0}^{2\pi}
R(\theta, \varphi, \psi)^{-1}
(- w_{1}\bm{v}_{1} - w_{2}\bm{v}_{2})^{\otimes 2}
R(\theta, \varphi, \psi)
\sin \theta d\theta d\varphi d\psi \\
& \left. \quad + 
\int_{0}^{\pi} \int_{3\pi/2}^{\varphi_{0} + 3\pi/2}  \int_{0}^{2\pi}
R(\theta, \varphi, \psi)^{-1}
(w_{1}\bm{v}_{1} - w_{2}\bm{v}_{2})^{\otimes 2}
R(\theta, \varphi, \psi)
\sin \theta d\theta d\varphi d\psi \right) \\ 
& = 
\frac{1}{4\pi^{2}} \left( 
\int_{0}^{\pi} \int_{\varphi_{0} - \pi/2}^{\pi/2} \int_{0}^{2\pi}
R(\theta, \varphi, \psi)^{-1}
\begin{pmatrix}
1+x \\
y \\
0 
\end{pmatrix}^{\otimes 2}
R(\theta, \varphi, \psi)
\sin \theta d\theta d\varphi d\psi \right. \\
& \left. \quad + 
\int_{0}^{\pi} \int_{\pi/2}^{\varphi_{0} + \pi/2}  \int_{0}^{2\pi}
R(\theta, \varphi, \psi)^{-1}
\begin{pmatrix}
-1+x \\
y \\
0 
\end{pmatrix}^{\otimes 2}
R(\theta, \varphi, \psi)
\sin \theta d\theta d\varphi d\psi \right). 
\end{align*}

Let $\alpha_{ij} \in \bbR$ for $1 \leq i, j \leq 2$. 
Suppose that $\alpha_{ij} = \alpha_{ji}$. 
Let 
\begin{align*}
(\beta_{ij})_{1 \leq i,j \leq 3}
& = \int 
\begin{pmatrix}
\cos \varphi & \sin \varphi & 0 \\
-\sin \varphi & \cos \varphi & 0 \\
0 & 0 & 1 
\end{pmatrix}
\begin{pmatrix}
\alpha_{11} & \alpha_{12} & 0 \\
\alpha_{21} & \alpha_{22} & 0 \\
0 & 0 & 0 
\end{pmatrix}
\begin{pmatrix}
\cos \varphi & -\sin \varphi & 0 \\
\sin \varphi & \cos \varphi & 0 \\
0 & 0 & 1 
\end{pmatrix}
d\varphi, \\
(\gamma_{ij})_{1 \leq i,j \leq 3}
& = \int_{0}^{\pi} 
\begin{pmatrix}
1 & 0 & 0 \\
0 & \cos \theta & \sin \theta \\
0 & -\sin \theta & \cos \theta 
\end{pmatrix}
(\beta_{ij})_{1 \leq i,j \leq 3}
\begin{pmatrix}
1 & 0 & 0 \\
0 & \cos \theta & -\sin \theta \\
0 & \sin \theta & \cos \theta 
\end{pmatrix}
\sin \theta d\theta, \\
(\delta_{ij})_{1 \leq i,j \leq 3}
& = \int_{0}^{2\pi} 
\begin{pmatrix}
\cos \psi & \sin \psi & 0 \\
-\sin \psi & \cos \psi & 0 \\
0 & 0 & 1 
\end{pmatrix}
(\gamma_{ij})_{1 \leq i,j \leq 3}
\begin{pmatrix}
\cos \psi & -\sin \psi & 0 \\
\sin \psi & \cos \psi & 0 \\
0 & 0 & 1 
\end{pmatrix}
d\psi. 
\end{align*}
Then 
\begin{align*}
\beta_{11} & = \frac{1}{4} \alpha_{11} (2\varphi + \sin 2\varphi) 
- \frac{1}{2} \alpha_{12} \cos 2\varphi + \frac{1}{4} \alpha_{22} (2\varphi - \sin 2\varphi), \\
\beta_{22} & = \frac{1}{4} \alpha_{11} (2\varphi - \sin 2\varphi) 
+ \frac{1}{2} \alpha_{12} \cos 2\varphi + \frac{1}{4} \alpha_{22} (2\varphi + \sin 2\varphi), \\
\beta_{33} & = 0, \\
\gamma_{11} & = 2\beta_{11}, \
\gamma_{22} = \frac{2}{3} \beta_{22} + \frac{4}{3} \beta_{33}, \
\gamma_{33} = \frac{4}{3} \beta_{22} + \frac{2}{3} \beta_{33}, \\
\delta_{11} & = \delta_{22} = \pi \gamma_{11} + \pi \gamma_{22}, \
\delta_{33} = 2\pi \gamma_{33}. 
\end{align*}
Hence
\[
2\delta_{11} - \delta_{22} - \delta_{33} = 2\pi (\beta_{11} - \beta_{22}) 
= 2\pi (\frac{1}{2} (\alpha_{11} - \alpha_{22}) \sin 2\varphi - \alpha_{12} \cos 2\varphi). 
\]
Since 
\[
\sin 2\varphi_{0} = \frac{2xy}{x^{2} + y^{2}} \quad \text{and} \quad 
\cos 2\varphi_{0} = \frac{x^{2} - y^{2}}{x^{2} + y^{2}}, 
\]
we have 
\begin{align*}
\delta & = 
\frac{1}{2\pi} \left(
\left[
\frac{1}{2} ((1+x)^{2} - y^{2}) \sin 2\varphi - (1+x)y \cos 2\varphi 
\right]_{\varphi_{0} - \pi/2}^{\pi/2} \right.\\
& \left. \quad + 
\left[
\frac{1}{2} ((-1+x)^{2} - y^{2}) \sin 2\varphi - (-1+x)y \cos 2\varphi 
\right]_{\pi/2}^{\varphi_{0} + \pi/2} \right) \\
& = \frac{1}{2\pi} (2x \sin 2\varphi_{0} + 2y (1 - \cos 2\varphi_{0})) \\
& = \frac{4y}{2\pi}  > 0. 
\end{align*}
\end{proof}

\section{Deviation angles}
\label{sec:deviation}

By Proposition~\ref{prop:deviation}, the permanent strain relates with 
the angle between the directions of stretching and splitting. 
In this section, we introduce the notion of deviation angles of a net 
based on this fact, in order to measure anisotropy. 

Non-crystalline (or amorphous) materials have no special directions. 
In other words, they are isotropic. 
Polycrystalline materials consisting of grains are also isotropic at the macro-scale, 
because anisotropies in grains are cancelled by averaging. 
This observation motivates the averaging in Section~\ref{sec:single}. 
In an isotropic elastomer, 
the polymer chains are uniformly distributed in all the directions. 
Such isotropy is usually assumed 
in classical theories of rubber elasticity as found in \cite{Treloar}. 
However, our model is applicable to anisotropic networks. 

Let $(X, \Phi)$ be an $N$-dimensional net. 
Let $E_{0} = \{ e_{1}, \iota (e_{1}), \dots, e_{n}, \iota (e_{n}) \}$ 
be a set of representatives of edges per period. 
We write $w_{i} = w(e_{i})$ and $\bm{v}_{i} = \Phi (e_{i})$ for $1 \leq i \leq n$. 
Take $\bm{u} \neq 0 \in \bbR^{N}$. 
Let $\epsilon_{i} \in \{\pm 1\}$ satisfy that $\bm{u} \cdot \epsilon_{i} \bm{v}_{i} \geq 0$. 
If $\bm{u} \cdot \bm{v}_{i} = 0$, we consider both of $\epsilon_{i} = \pm 1$. 
Let $\bm{z} = \sum_{i=1}^{n} w_{i} \epsilon_{i} \bm{v}_{i}$, 
which is multivalued for $\bm{u}$. 
We call the angle between $\bm{u}$ and $\bm{z}$ 
the \emph{deviation angle}. 
It is at most $\pi/2$. 
The \emph{maximal (resp. minimal) deviation angle} for $(X, \Phi)$ 
is the maximum (resp. minimum) of the deviation angles over the vectors $\bm{u}$ and $\bm{z}$. 
These are independent of a choice of a period. 

A net in our definition is never isotropic. 
Nonetheless, we may imagine an isotropic network consisting of uniformly distributed edges. 
In this case, the maximal deviation angle is equal to zero. 
Thus the maximal deviation angle can be regarded as a measure of anisotropy.

\begin{prop}
\label{prop:maxdev}
The maximal deviation angle is positive and attained 
when $\bm{u} \cdot \bm{v}_{i} = 0$ for some $i$. 
Furthermore, the maximum is an accumulation point of the deviation angles on neighborhood of such $\bm{u}$. 
\end{prop}
\begin{proof}
There are only finitely many possibilities of $\bm{z}$. 
For fixed $\bm{z}$, 
the subspace of $S^{N-1}$ consisting of the corresponding $\bm{u}$'s 
is compact. 
Hence the maximum of the angle between $\bm{u}$ and $\bm{z}$ exists. 

Suppose that $\bm{u} \cdot \bm{v}_{i} \neq 0$ for any $i$. 
Consider small perturbation of $\bm{u}$. 
Since $\bm{z}$ is constant, 
the deviation angle can be larger. 
Hence the maximal deviation angle is positive and not attained by such $\bm{u}$. 

Let $\bm{u}_{\mathrm{max}}$ and $\bm{z}_{\mathrm{max}}$ be 
the vectors that attain the maximal deviation angle. 
We may assume that 
$\bm{u}_{\mathrm{max}} \cdot \bm{v}_{1} = \dots = \bm{u}_{\mathrm{max}} \cdot \bm{v}_{m} = 0$ and 
$\bm{u}_{\mathrm{max}} \cdot \bm{v}_{m+1} > 0, \dots, \bm{u}_{\mathrm{max}} \cdot \bm{v}_{n} > 0$. 
Then $\bm{z}_{\mathrm{max}} = \sum_{i=1}^{n} w_{i} \bm{v}_{i}$. 
If $\bm{z}_{\mathrm{max}} \cdot \bm{v}_{i} < 0$ for some $1 \leq i \leq m$, 
we obtain a larger deviation angle by replacing $\bm{v}_{i}$ with $-\bm{v}_{i}$. 
Hence $\bm{z}_{\mathrm{max}} \cdot \bm{v}_{i} \geq 0$ for any $1 \leq i \leq m$. 
The vector $\bm{u} = \bm{u}_{\mathrm{max}} + r \bm{z}_{\mathrm{max}}$ 
induces the same $\bm{z} = \bm{z}_{\mathrm{max}}$ for any small $r > 0$. 
Therefore the maximal deviation angle is an accumulation point. 
\end{proof}

Representation surfaces of Young's modulus are used as a method to evaluate the anisotropy of monocrystallines and lattice structures~\cite{nye1985physical,XU2016}.
The surface is a spherical surface, with the length of each radius vector being the reciprocal of Young's modulus corresponding to the direction.
However, Young's modulus is not effective for the model in this paper.
(For example, Young's modulus of a standard net will be the same for all directions~\cite{KY}.)
As mentioned above, the deviation angle can take its place in our model, so we use representation surfaces of deviation angles as an analogy.
Namely, the surface is obtained by the length of each unit vector representing the sphere as the supremum of the deviation angles to the direction.

\begin{ex}\label{ex:deviation-bcc}
Let us consider a 3-dimensional net obtained by the body-centered cubic (bcc) lattice as follows: 
\begin{itemize}
\item The image of the period homomorphism is $\bbZ^{3}$. 
\item The points $(0,0,0)$ and $(1/2, 1/2, 1/2)$ are representatives of the vertices. 
\item The edges join the first and second nearest vertices with distances $\sqrt{3}/2$ and $1$. 
\item The weights are equal to one. 
\end{itemize}
Then the deviation angles are indicated in Figure~\ref{fig:devbcc}. 
The maximal deviation angle is equal to $\arccos \sqrt{2/3} \sim 0.615$ 
and the minimal deviation angle is equal to zero given by $\bm{u} = (0, 1, 1)$ and $\bm{u} = (1, 1, 1)$, respectively. 
\end{ex}

\begin{figure}[htb]
\begin{minipage}[b]{0.56\textwidth}
    \centering
    \includegraphics[width=\linewidth]{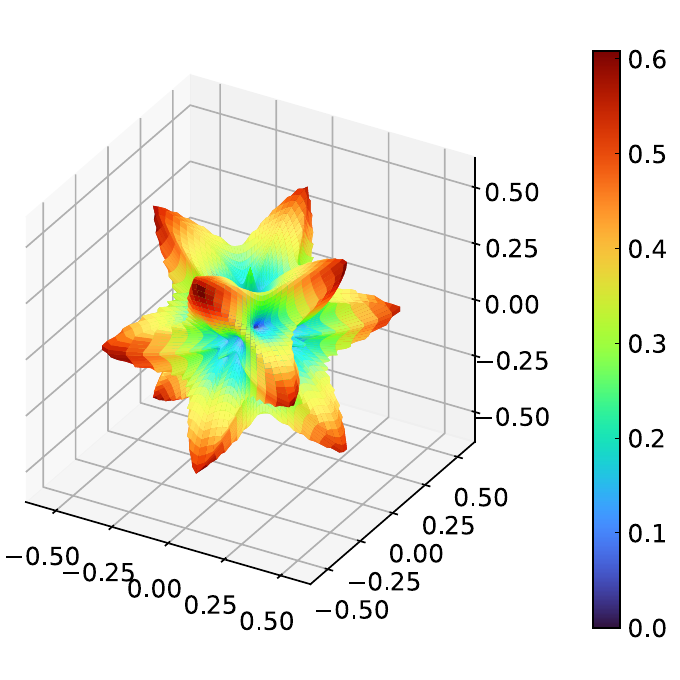}\\
    (a)
\end{minipage}
\begin{minipage}[b]{0.40\textwidth}
    \centering
    \includegraphics[width=\linewidth]{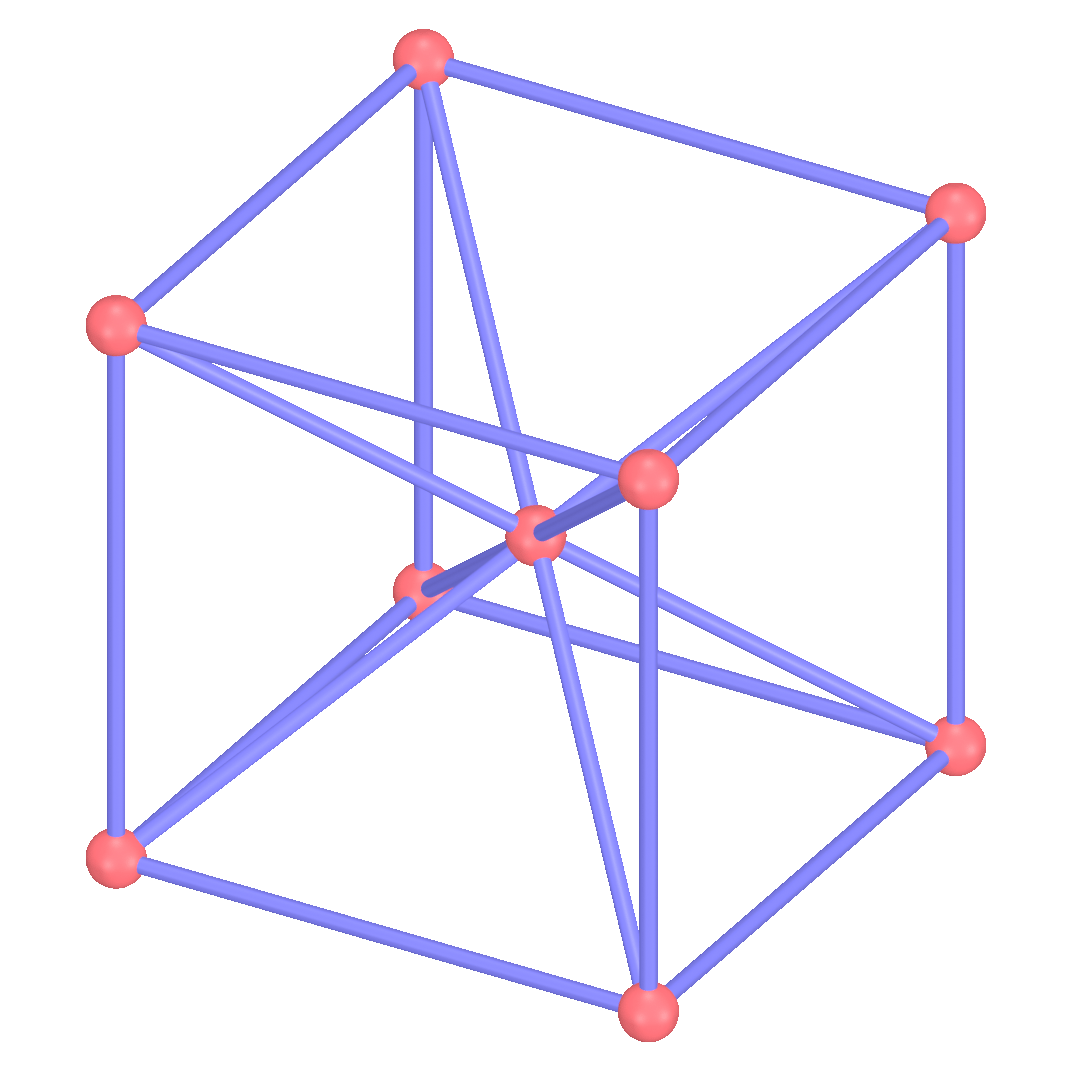}\\
    (b)
\end{minipage}
\caption{(a) The representation surface of the bcc lattice described in Example~\ref{ex:deviation-bcc}. (b) The bcc lattice used to obtain the surface in~(a). The representation surface and the lattice are viewed from the same angle.}
\label{fig:devbcc}
\end{figure}

The minimal deviation angle is equal to zero 
if and only if there exists $\bm{u} \neq 0 \in \bbR^{N}$ such that 
$\bm{u} = \bm{z}$ for $\bm{z}$ as above. 
This is very likely to hold, 
but we show it only for the 2-dimensional nets. 

\begin{thm}
If $N=2$, the minimal deviation angle is equal to zero. 
\end{thm}
\begin{proof}
We may assume that 
$\bm{v}_{i} = (r_{i}\cos \theta_{i}, r_{i}\sin \theta_{i})^{\tp}$, 
$r_{i} > 0$, $0 = \theta _{1} < \theta_{2} < \dots < \theta_{n} < \pi$, 
and $w_{i} = 1$. 
Let $\bm{z}_{i} = - \bm{v}_{1} - \dots - \bm{v}_{i-1} + \bm{v}_{i} + \dots + \bm{v}_{n} 
= \sum_{j=1}^{n} \epsilon_{ij} \bm{v}_{j}$, 
where $\epsilon_{ij} = -1$ for $i > j$ and $\epsilon_{ij} = 1$ for $i \leq j$. 

Assume that the minimal deviation angle is not equal to zero. 
In other words, $\bm{u} \neq \bm{z}_{i}$. 
Then for any $i$, it does not hold that $\bm{z}_{i} \cdot (\epsilon_{ij} \bm{v}_{j}) \geq 0$ for all $j$. 
Hence $\bm{z}_{1} \cdot \bm{v}_{1} < 0$ or $\bm{z}_{1} \cdot \bm{v}_{n} < 0$. 
We may assume that $\bm{z}_{1} \cdot \bm{v}_{1} < 0$. 
Then inductively $\bm{z}_{i} \cdot \bm{v}_{i} < 0$ for any $i$. 
On the other hand, 
$\sum_{i=1}^{n} \bm{z}_{i} \cdot \bm{v}_{i} = \sum_{i=1}^{n} \|\bm{v}_{i}\|^{2} > 0$. 
This is a contradiction. 
\end{proof}

\begin{conj}
The minimal deviation angle is equal to zero for any $N$. 
\end{conj}

\section{Making new nets by blending polymers}
\label{sec:blend}

In this section, we discuss building new nets by blending several types of triblock copolymers.
We can obtain nets with high local anisotropy such as the kagome lattice.

We note the terminology. 
We use several simple nets with customary names such as the kagome lattice.

\subsection{Blending several types of block copolymers}
Here we consider a mathematical model of constructing new nets by blending several types of triblock copolymers of the same length.
We start with the case where the triblock copolymers are of types ABA and ABC.
Monomers A, B, and C are incompatible; moreover, monomers A and C form hard domains of sphere structure, individually.
In this subsection, we assume that the weight of each edge is equal.

First, we deal with 2-dimensional nets.
We assume that the hard domains of monomer A and monomer C form the triangular lattice.
In Figure \ref{fig:networks-blending}, red vertices are domains of A monomers and orange vertices are those of C monomers.

\begin{ex}\label{ex:2dblend}
Suppose that the subnet formed by ABA copolymers is the kagome lattice as in Figure \ref{fig:networks-blending}(a).
Then we can see that the ratio of the number of triblock copolymers is
ABA : ABC$=1:1$.
Next, suppose that the subnet of ABA copolymers is the hexagonal lattice as in Figure \ref{fig:networks-blending}(b).
Then we can see that the ratio of the number of triblock copolymers is
ABA : ABC$=1:2$.
\end{ex}

The kagome lattice and the hexagonal lattice are harmonic subnets of the triangular lattice.
See Theorem \ref{thm:subnet_2dim}.

\begin{figure}[htbp]
\begin{minipage}{0.48\hsize}
\centering
\includegraphics[width=0.95\textwidth]{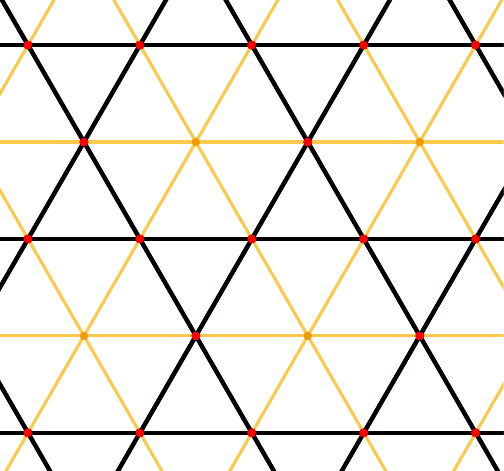}\\
(a)
\end{minipage}
\hfill
\begin{minipage}{0.48\hsize}
\centering
\includegraphics[width=0.95\textwidth]{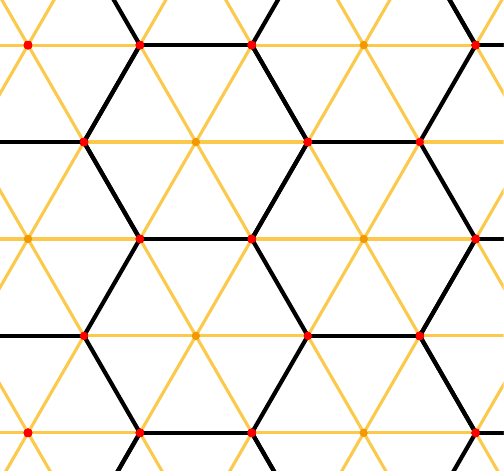}\\
(b)
\end{minipage}

\caption{Example of 2-dimensional nets.
(a) and (b) illustrate the kagome lattice and the hexagonal lattice made by blending, respectively. 
Red vertices are monomer A domains and orange vertices are monomer C domains.
Black edges indicate bridges of ABA copolymers and orange edges bridges of ABC copolymers.
The kagome lattice is a mixture of ABA and ABC copolymers in a $1:1$ ratio, while the hexagonal lattice is in a $1:2$ ratio.}
\label{fig:networks-blending}
\end{figure}

Next, we consider 3-dimensional nets.
Here, we discuss the blend of three types of triblock copolymers of types ABA, ABC, and CBC.
We assume that the hard domains of monomer A and monomer C form the body-centered cubic (bcc) lattice.
In Figure 
\ref{fig:dia-blending}, red vertices are domains of A monomers and orange vertices are those of C monomers.

\begin{ex}\label{ex:3dblend}
Suppose that the subnets of ABA copolymers CBC copolymers form the 2dia lattice as in Figure~\ref{fig:dia-blending}
\cite{RCSR}.
In this case, the subnet of ABC copolymers also forms the 2dia net.
Then, the ratio of the number of triblock copolymers is $\text{ABA} : \text{ABC} : \text{CBC} = 1:2:1$.
\end{ex}

The 2dia lattice is an example of harmonic subnets of the bcc lattice.
See Theorem~\ref{thm:subnet_3dim}.

\begin{figure}[htbp]
\begin{minipage}{0.48\hsize}
\centering
\includegraphics[width=0.9\textwidth]{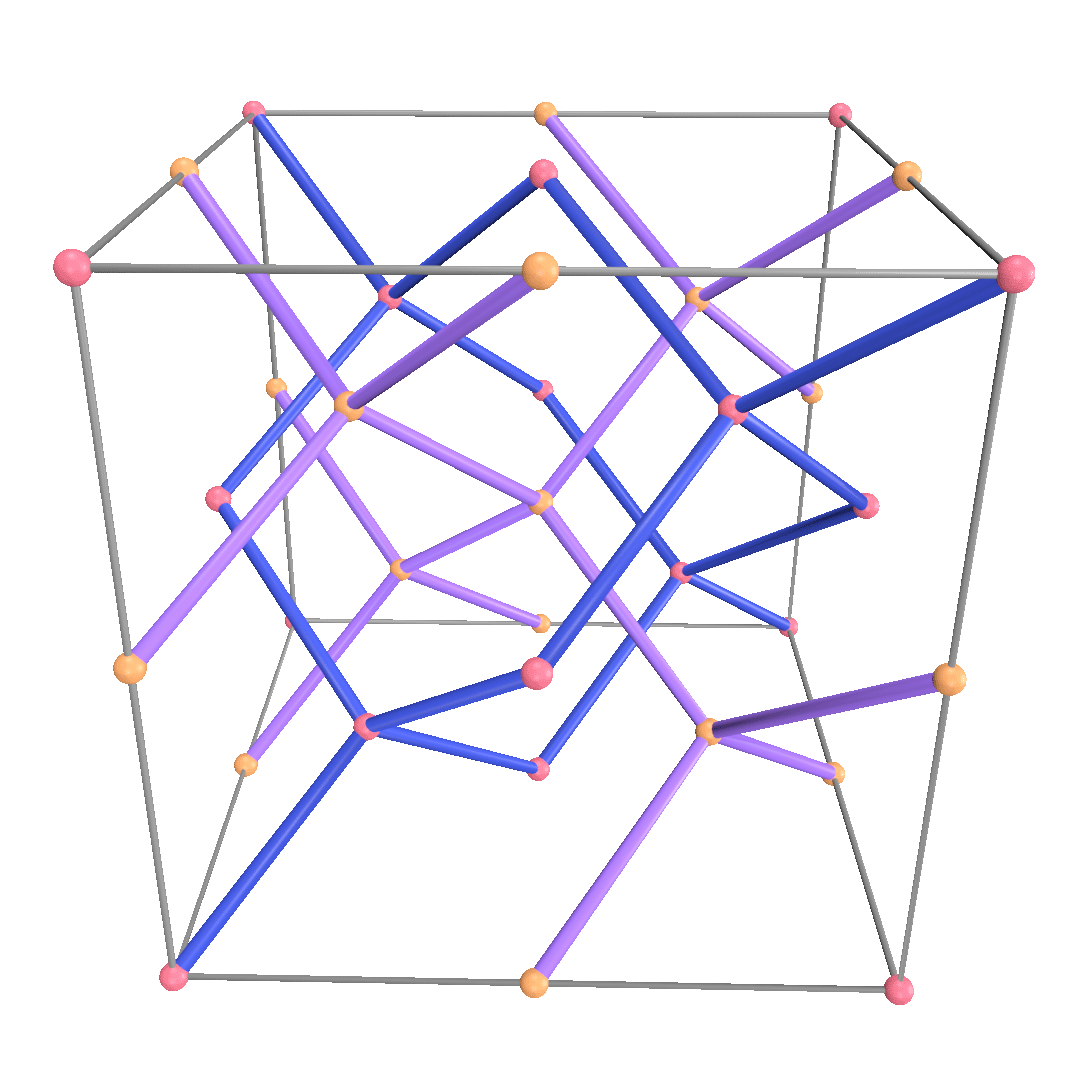}\\
(a)
\end{minipage}
\hfill
\begin{minipage}{0.48\hsize}
\centering
\includegraphics[width=0.9\textwidth]{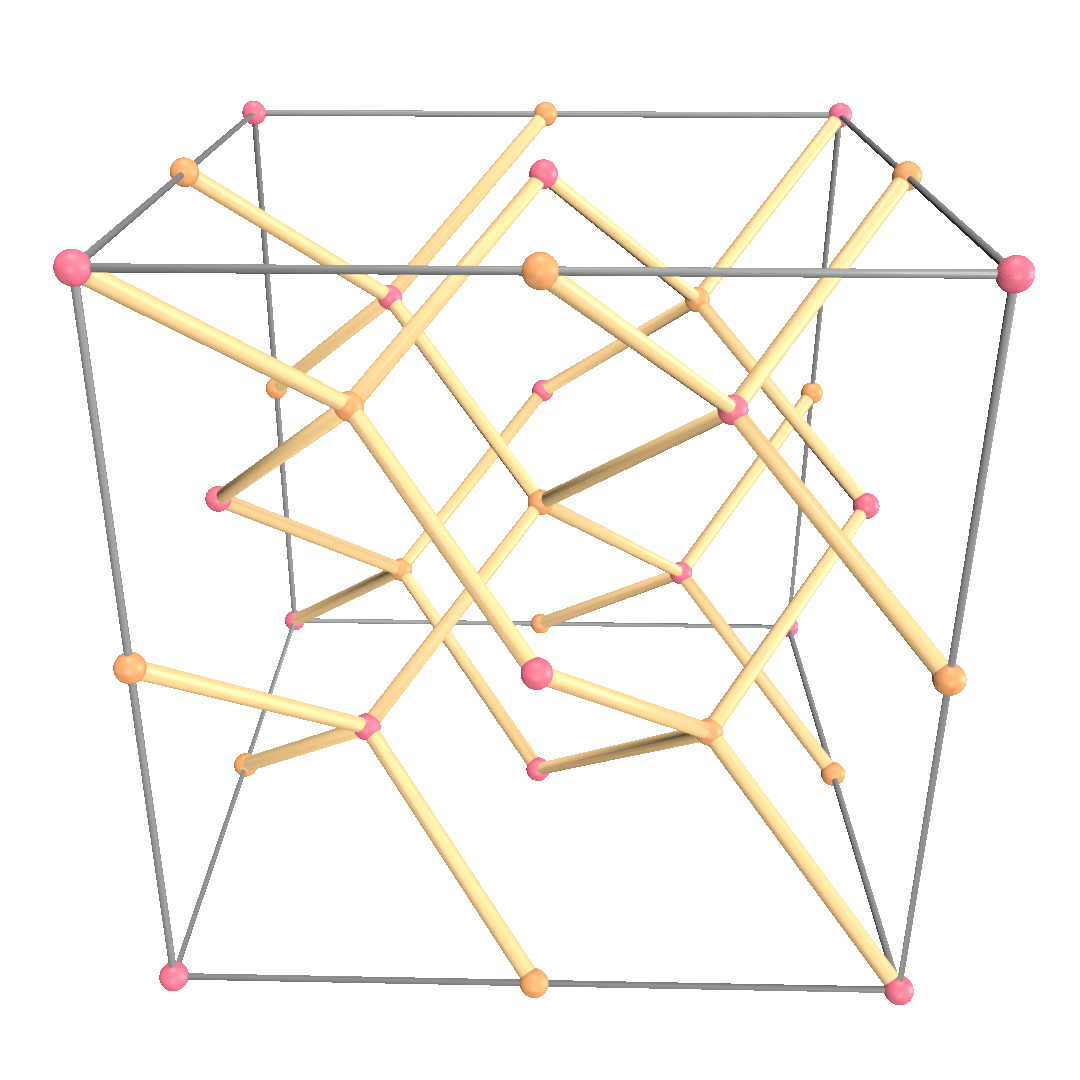}\\
(b)
\end{minipage}

\caption{Example of 3-dimensional nets.
(a) The 2dia lattice is made of the subsets of ABA copolymers and CBC copolymers included in the bcc lattice.
Note that each component is the dia lattice. (b) The subnet of ABC copolymers contains the edges connecting the two components of the 2dia lattice.}
\label{fig:dia-blending}
\end{figure}

\begin{ques}
    Can the subnets in Examples~\ref{ex:2dblend} and \ref{ex:3dblend} be created by simulation or actual synthesis?
\end{ques}

\subsection{Characterization of harmonic subnets}
In this subsection, we characterize harmonic subnets of the triangular lattice and the bcc lattice.

We use the following assumption for our model.

\begin{as}
Assume that a net is composed of triblock copolymers of types ABA, ABC, and CBC.
When the monomer A domain is much more robust than the monomer C domain, we can assume that the stress chain consists of bridges of ABA triblock copolymers.  
\end{as}

Based on this assumption, we deal with subnets of ABA-type triblock copolymers.
We consider the case where this subnet constitutes a stress chain and is also harmonic when considering only this subnet.

Harmonicity and absence of type CBC induce strong restrictions on the network as follows.

\begin{thm}[2-dimensional case]\label{thm:subnet_2dim}
Consider the standard triangular lattice $X$, 
whose edges join the first proximity points. 
Suppose that the vertices of $X$ 
are divided into two types A and C, 
and no pair of vertices of type C are adjacent. 
Let $Y$ be the net consisting of the vertices of type A and the edges between them. 
Suppose that the net $Y$ is harmonic. 
Then $Y$ is the hexagonal lattice or the kagome lattice. 
\end{thm}
\begin{proof}
Since $Y$ is harmonic, 
the edge configuration around a vertex of type A has three possibilities 
as indicated in Figure~\ref{fig:2-dim_harmonic_2poly}. 
The first case is excluded 
since no pair of vertices of type C are adjacent. 
If a vertex of type A has precisely three adjacent vertices of type A, 
the configurations around the adjacent vertices are determined. 
By continuing this argument, we obtain that $Y$ is the hexagonal lattice. 
Similarly, if a vertex of type A has precisely four adjacent vertices of type A, 
then $Y$ is the kagome lattice. 
\end{proof}

\begin{figure}[htbp]
    \begin{minipage}[b]{0.32\textwidth}
    \centering
    \includegraphics[width=0.70\textwidth,page=1]{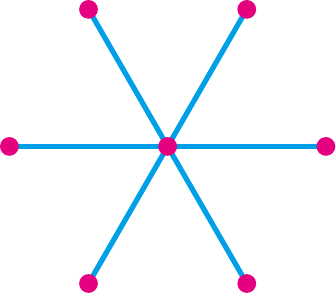}\\
    (a)
    \end{minipage}
    \begin{minipage}[b]{0.32\textwidth}
    \centering
    \includegraphics[width=0.70\textwidth,page=3]{./figs/edge_configurations}\\
    (b)
    \end{minipage}
    \begin{minipage}[b]{0.32\textwidth}
    \centering
    \includegraphics[width=0.70\textwidth,page=2]{./figs/edge_configurations}\\
    (c)
    \end{minipage}
    \caption{The neighbor of each vertex is either (a), (b), or, (c), under the assumption of Theorem~\ref{thm:subnet_2dim}.
    Here the blue edges represent the bridges of ABA triblock copolymers.}
    \label{fig:2-dim_harmonic_2poly}
\end{figure}

The same argument shows the following. 

\begin{thm}[3-dimensional case]\label{thm:subnet_3dim}
Consider the standard body-centered cubic lattice $X$, 
whose edges join the first proximity points. 
Suppose that the vertices of $X$ 
are divided into two types A and C, 
and no pair of vertices of type C are adjacent. 
Let $Y$ be the net consisting of the vertices of type A and the edges between them. 
Suppose that $Y$ is harmonic. 
Then $Y$ is one of the following lattices: 
\begin{enumerate}[label=\textup{(\roman*)}]
\item the bcc lattice of twice the size, 
\item the complement of the bcc lattice of twice the size, and 
\item a lattice obtained by gluing blocks (a) and (b) in Figure~\ref{fig:3-dim_harmonic_2poly}. 
\end{enumerate}
In the last case, 
if there is a block (a), then four adjacent, and thus planarly placed, blocks (a) are determined. 
For any sequence of (a) and (b), 
there is a unique desired net. 
\end{thm}

\begin{figure}[htbp]
    \begin{minipage}[b]{0.20\textwidth}
    \centering
    \includegraphics[width=0.98\textwidth]{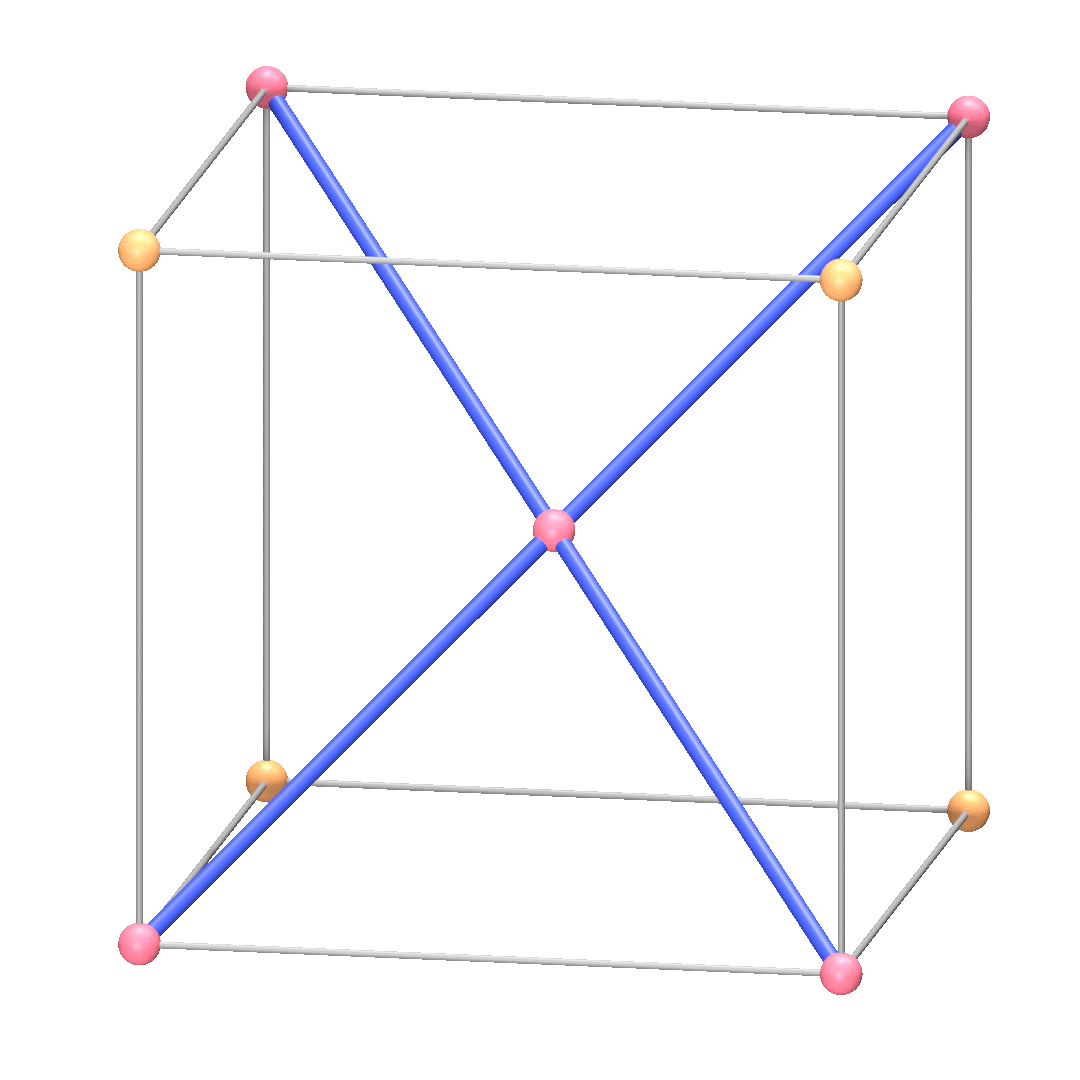}\\
    (a)\\
    
    \includegraphics[width=0.98\textwidth]{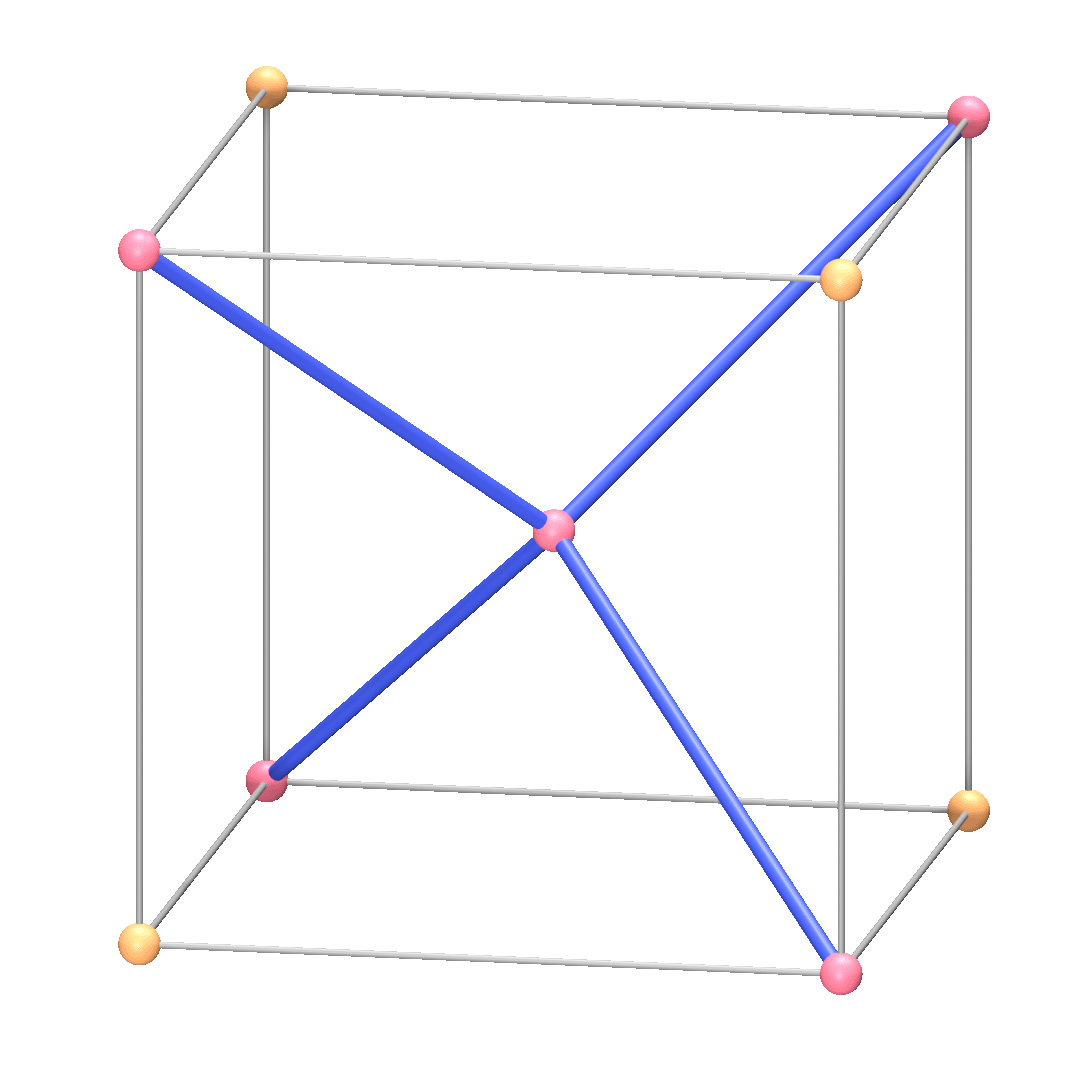}\\
    (b)
    \end{minipage}
    \begin{minipage}[b]{0.39\textwidth}
    \centering
    \includegraphics[width=1.0\textwidth]{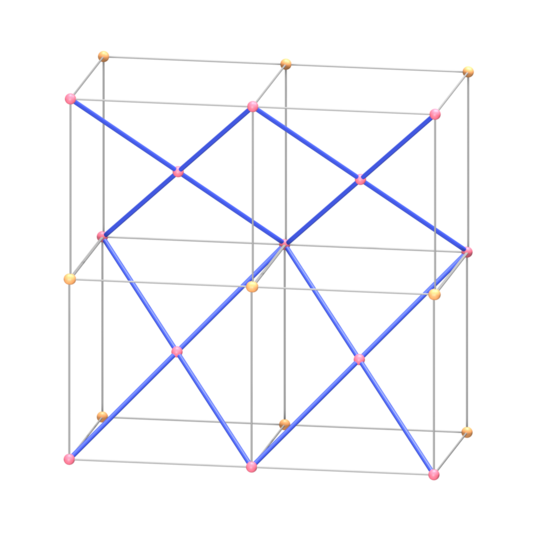}\\
    (c)
    \end{minipage}
    \begin{minipage}[b]{0.39\textwidth}
    \centering
    \includegraphics[width=1.0\textwidth]{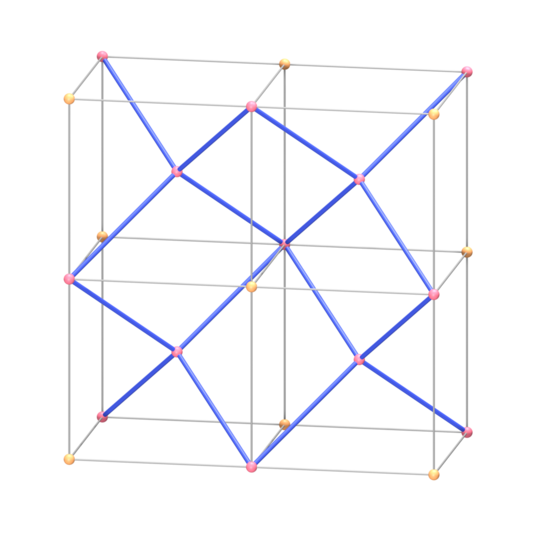}\\
    (d)
    \end{minipage}
    \caption{(a) and (b) are basic building blocks. Only the edges corresponding to ABA triblock copolymers are shown here. (c) and (d) are nets consisting of four building blocks shown in (a) and (b), respectively. In general, any net that combines (a) and (b) (and their mirror images)
    so that the vertex types match is allowed.  \label{fig:3-dim_harmonic_2poly}}
\end{figure}

\section{Simulation and comparison of several nets}
\label{sec:simulation}

In this section, we simulate the extension of several nets and compare their properties 
reflecting the anisotropic nature.

\subsection{Differences between the net extension model in this paper and the Kodama-Yoshida model}

We describe the mathematical model of net extension used in this paper.
This model is a slightly modified version of the Kodama-Yoshida model, called a \emph{slow deformation}, introduced in~\cite{KY}.
By applying a suitable rotation to a net in $\mathbb{R}^N$, an extension of the net can be represented by a transformation by $A(\lambda) = \mathrm{diag}(\lambda, \lambda^{-1/(N-1)}, \dots, \lambda^{-1/(N -1)}) \in \mathrm{SL}(N, \mathbb{R})$.

We first describe the setting of nets for the model.
We use a standard realization of each net as an initial state.
Especially, we assume a `normalized' one, to be introduced in the next subsection, for comparing properties affected by the combinatorial structure of each net.
We also assume that all edge weights of a net are identical.
To simplify the situation, the weights of the edges that arise after vertex splits are the same as those given in the initial state, independent of the loop chains adjacent to the split vertex.

Next, we discuss modeling the behavior of the expansion of a net.
The strain at each vertex increases with the extension of a net, and by considering the maximum eigenvalue of the local tension tensor at a vertex as the strain, the splitting of the structure (e.g., the hard domain of the TPE) can be modeled.
Note that, in the original slow deformation, the vertices of the same strain split simultaneously, but in this paper, we assume that only one vertex split at a time.
Since we compare the properties of normalized nets, we assume that the threshold of the maximum eigenvalue on which the vertex splitting is based depends on the number of vertices in the unit cell of each normalized net in the initial state.
If the energy of a net per period is fixed, the average eigenvalue per vertex is approximately inversely proportional to the number of vertices.
Thus, for example, when comparing two normalized nets with $n$ and $m$ vertices in the unit cell, if the threshold for vertex splitting in the former net is $s$, the threshold for the latter is $s \times n/m$.
Suppose the edges adjacent to the split vertex are orthogonal to the divided direction given by the eigenvector for the maximal eigenvalue.
In that case, we assume that two edges derived from the original edge are adjacent to each vertex after the split.
The Kodama-Yoshida model considers contractions of vertices, but for simplicity, we do not consider them in this paper.
Furthermore, since the degree of the vertices of the initial state of each net considered in this paper is all equal, we assume that the threshold for determining vertex splitting is independent of the degree.

\subsection{Normalized nets}

This subsection introduces a method of normalization to compare the behavior of extended periodic realized nets.
It is difficult to compare the properties of periodic realizations of two given nets because their energies and tension tensors may differ even if the nets have the same shape.
Therefore, it is necessary to select a suitable realization among the periodic realizations and use it for comparison.
We discuss the `normalization' of nets to solve these problems.

First, let us precisely define two periodic realizations to be `similar'.
Let $X$ be a net with a vertex set $V$, $w_1$ and $w_2$ be weight functions on $X$, and $\Phi_1$ and $\Phi_2$ be periodic realizations of $(X, w_1)$ and $(X, w_2)$, respectively.
We say that the realizations $\Phi_1$ and $\Phi_2$ are \emph{similar} if there exists a similarity transformation $f$ (i.e., the composition of an isometry and a uniform scaling) of $\bbR ^N$ and a positive number $k$ such that $f \circ \Phi_1 = \Phi_2$ holds on $V$, and $w_2(e) = k w_1(e)$ for each edge $e$.
Note that the above definition requires using the same period for two similar realizations.

For two similar periodic realizations, the tension tensor of one realization is conjugate to the matrix that is a constant multiple of the other one by an orthogonal matrix.
Suppose that a periodic realization $\Phi_1$ is similar to one $\Phi_2$ by a similarity transformation $f$.
Since $f$ is the composition of an isometry and a uniform scaling, it can be represented by a specific affine transformation $f(\bm{x}) = s R\bm{x} + \bm{c}$ for any $\bm{x} \in \bbR^N$, where $R$ is an orthogonal matrix, $\bm{c}$ is a vector of $\bbR^N$, and $s$ is a scale factor.
A Euclidean translation does not change the vector corresponding to each edge.
Additionally, since $\bm{x}^{\otimes 2} = \bm{x} \cdot \bm{x}^\top$, for the transformed vector $sR\bm{x}$, $(sR\bm{x})^{\otimes 2} = s^2 R \bm{x} \cdot (R \bm{x})^\top$.
Hence, for each vertex $v$, we have $\mathcal{T}_2(v) = k s^2 R \mathcal{T}_1(v) R^\top$, where $\mathcal{T}_i(v)$ is the local tension tensor at $v$ for $\Phi_i$, and $k$ is a positive number adjusting weights.
So, the global tension tensor $\mathcal{T}(X, \Phi_2)$ is also equal to $k s^2 R \mathcal{T}(X, \Phi_1) R^\top$.

For a given net $(X, \Phi)$, there is a unique periodic realization $\Phi_0$ similar to $\Phi$ such that the energy and the volume per period are equal to $1$.
We call $\Phi_0$ the \emph{normalized} (periodic) realization of $\Phi$, and the pair $(X, \Phi_0)$ the \emph{normalized} (periodic) net of the given net.
To normalize, we first scale and set the volume to $1$, then adjust the weights to set the energy to $1$.
Note that the normalized realizations depend on the period homomorphisms of nets, even if those are of the same periodic realization (as a mapping).

For given two normalized nets $(X, \Phi_1)$, $(X, \Phi_2)$ and two actions $L_1$, $L_2$ on $X$, the tension tensor of the normalized realization $\Phi_i$ for each of the actions is also obtained by multiplying a constant with the other.
However, the given periodic realization must be injective, and the image of each actions by the period homomorphism need to span $\bbR^N$.
Suppose a periodic realization and its period homomorphism satisfy the two conditions.
In that case, the `period' on the image of $X$ under the periodic realization conversely derives a new $\mathbb{Z}^N$-action on $X$.
We can see by using the idea of constructing a Wigner-Seitz cell~\cite{WS} that we have a $\mathbb{Z}^N$-action $L_0$ that contains the whole period of $X$ as a subgroup.
So, $L_1$ and $L_2$ are finite index subgroups of $L_0$.
Hence, it holds that $\mathcal{T}(X, \Phi_2) = [L_0 : L_2] / [L_0 : L_1] \mathcal{T}(X, \Phi_1)$, where $[L_0 : L_i]$ is the index.

We have the following proposition.
\begin{prop}
    Let $\Phi$ be a periodic realization of a net $X$ in $\bbR^N$.
    The tension tensor of the normalized periodic realization of $\Phi$ is unique up to conjugation by an orthogonal matrix.
    In particular, if $\Phi$ is a standard realization, then the tension tensor of the normalized realization is equal to $\frac{1}{N} I_N$, where $I_N$ is the identity matrix.
\end{prop}

\subsection{2-dimensional harmonic subnets}

In this subsection, we simulate extensions of nets obtained by Theorem~\ref{thm:subnet_2dim} and other well-known nets in Figure \ref{fig:2d_nets} and compare their properties.

Figure~\ref{fig:2d_graphs} shows the averages of stress and maximum eigenvalue of each normalized standardly realized net when stretched to equally spaced twelve directions.
The vertices were assumed to split when the maximum eigenvalue reached $5$ ($20/3$ for the kagome lattice).

In the four nets simulated, the triangular, square, and hexagonal lattices were found to have similar properties (see Figures~\ref{fig:2d_graphs}(a--c)).
In particular, triangular and square lattices have almost the same mean maximum eigenvalue.
Hence, by extending, vertex splittings of the two nets occur almost simultaneously.
In addition, the stresses of these nets are also almost the same.
The timing of the first and second vertex splits is different for the hexagonal lattice.
However, the third splitting is almost the same as the other nets.
Interestingly, even after the third splitting of the hexagonal lattice, its stress curve is almost the same as that after the second splitting of the triangular lattice.
Therefore, these three nets have similar properties.
In detail, however, we should note that the hexagonal lattice can store the most energy of these nets.

The kagome lattice has different properties to the rest of the nets simulated (see Figure~\ref{fig:2d_graphs}(d)).
After the first splitting of the kagome lattice, its stress curve is already almost the same as that of each other nets after two splittings.
Hence, it can be seen that the kagome lattice can not store as much as in the others.
However, as for the splitting, the first splitting occurs earlier than the other nets, but the remaining two splits occur later than the others.
Thus, the first splitting causes a significant change in properties, but later changes occur more slowly.

The kagome lattice is also unique for the average permanent strains (see Table~\ref{tab:mean_permanent_strain}).
For the triangular, square, and hexagonal lattices, the triangular lattice is slightly larger, but the values are almost the same after each splitting.
However, the average permanent strain of the kagome lattice is larger than that of the other nets.
The difference is especially large at the first split.
This also shows the effect of the first vertex splitting on properties.

\begin{figure}[htbp]
\begin{minipage}[c]{0.243\linewidth}
\centering
\includegraphics[width=.95\linewidth]{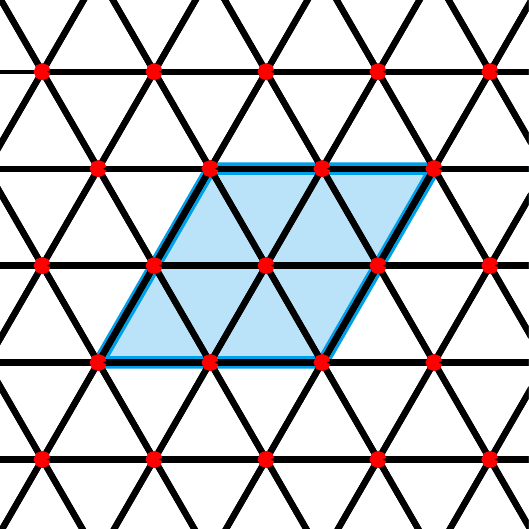}\\
(a)
\end{minipage}
\begin{minipage}[c]{0.243\linewidth}
\centering
\includegraphics[width=.95\linewidth]{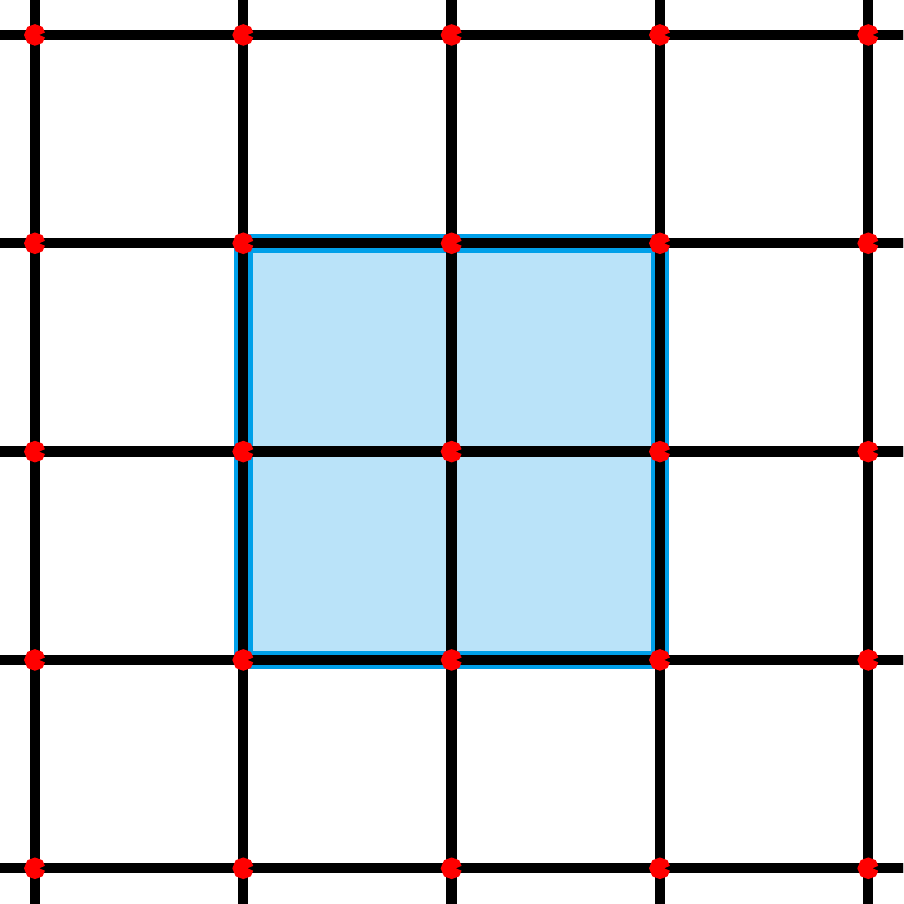}\\
(b)
\end{minipage}
\begin{minipage}[c]{0.243\linewidth}
\centering
\includegraphics[width=.95\linewidth]{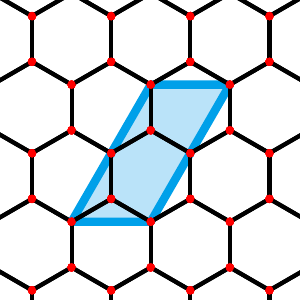}\\
(c)
\end{minipage}
\begin{minipage}[c]{0.243\linewidth}
\centering
\includegraphics[width=.95\linewidth]{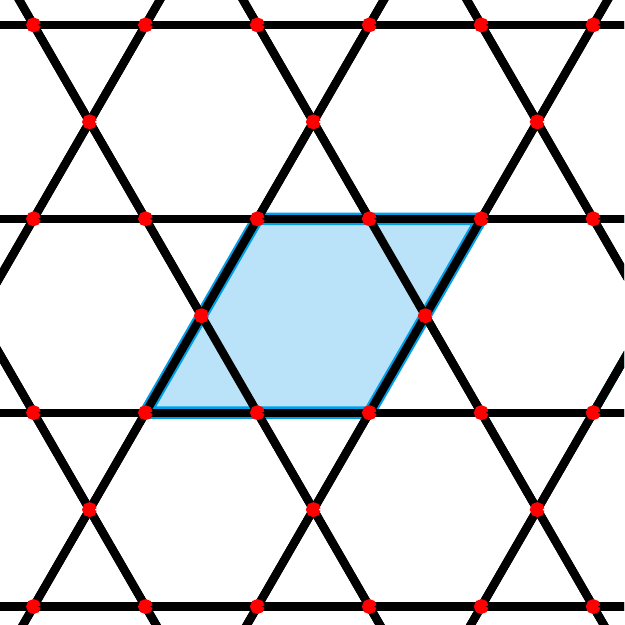}\\
(d)
\end{minipage}
\caption{Well-known 2d-nets and their periods. (a)~The triangular lattice. (b) ~The square lattice. (c)~The hexagonal lattice. (d)~The kagome lattice. \label{fig:2d_nets}}

\end{figure}

\begin{figure}[htbp]
\addtocounter{figure}{1}

\begin{center}
\begin{minipage}[b]{0.49\linewidth}
\centering
\includegraphics[width=\linewidth]{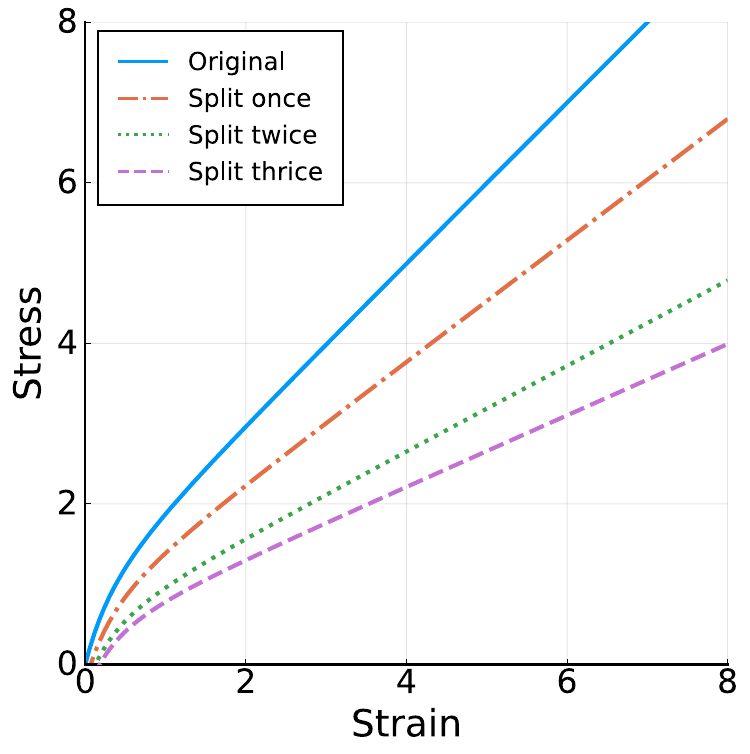}
\end{minipage}
\begin{minipage}[b]{0.01\linewidth}
\end{minipage}
\begin{minipage}[b]{0.49\linewidth}
\centering
\includegraphics[width=\linewidth]{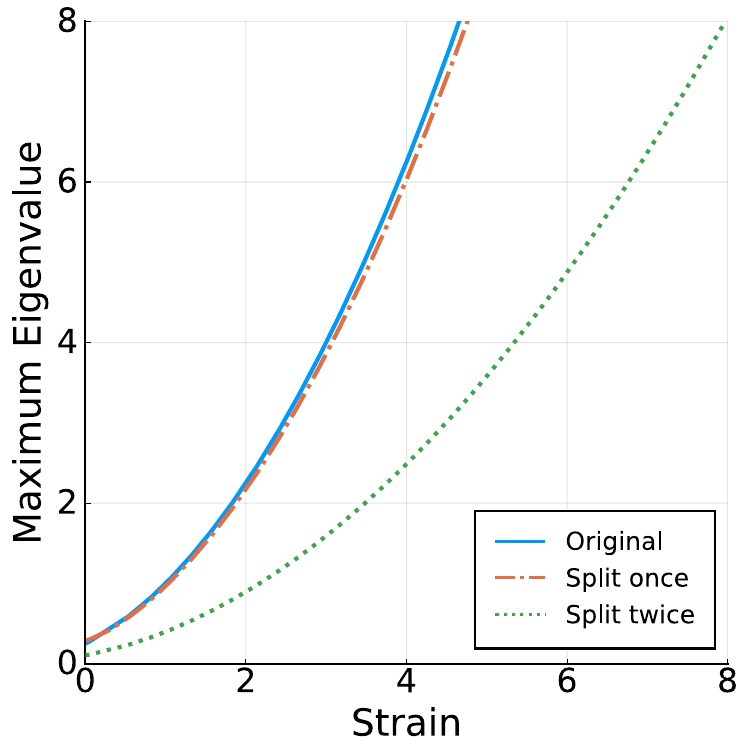}
\end{minipage}

(a) The triangular lattice.
\end{center}

\begin{center}
\begin{minipage}[b]{0.49\linewidth}
\centering
\includegraphics[width=\linewidth]{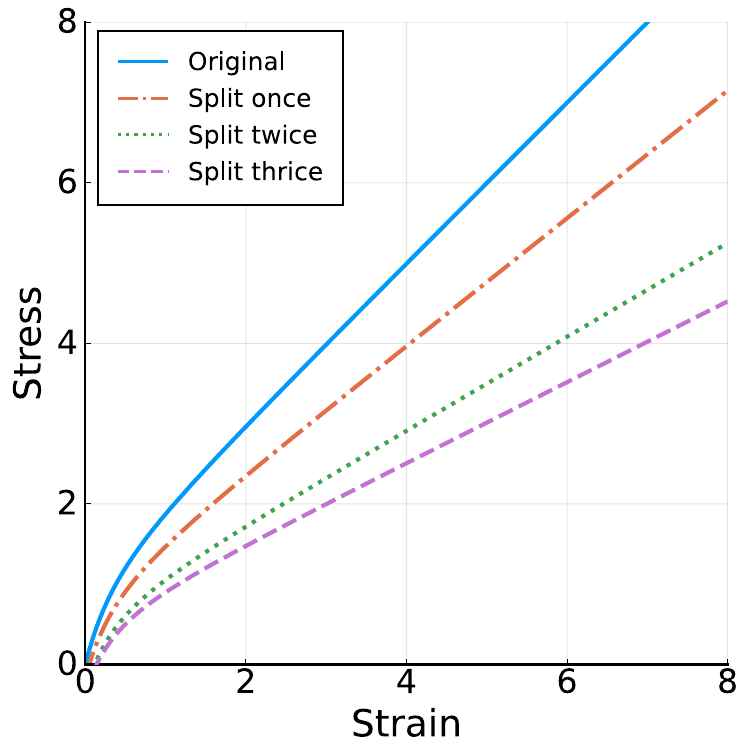}
\end{minipage}
\begin{minipage}[b]{0.01\linewidth}
\end{minipage}
\begin{minipage}[b]{0.49\linewidth}
\centering
\includegraphics[width=\linewidth]{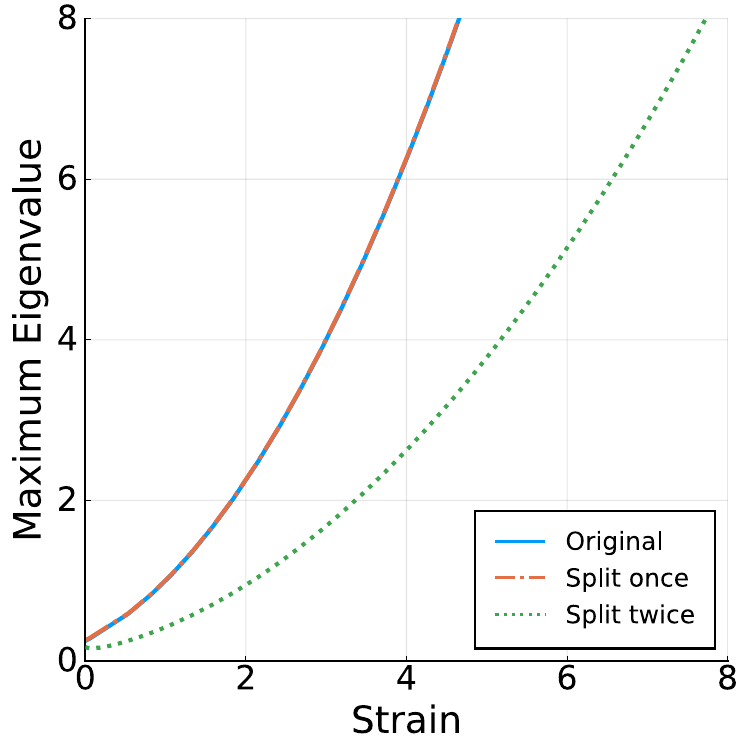}
\end{minipage}

(b) The square lattice.
\end{center}

\begin{flushleft}
\textsf{\bfseries\small Part of Figure~\ref{fig:2d_graphs}.}
\end{flushleft}

\end{figure}

\begin{figure}[htbp]
\addtocounter{figure}{-1}

\begin{center}
\begin{minipage}[b]{0.49\linewidth}
\centering
\includegraphics[width=\linewidth]{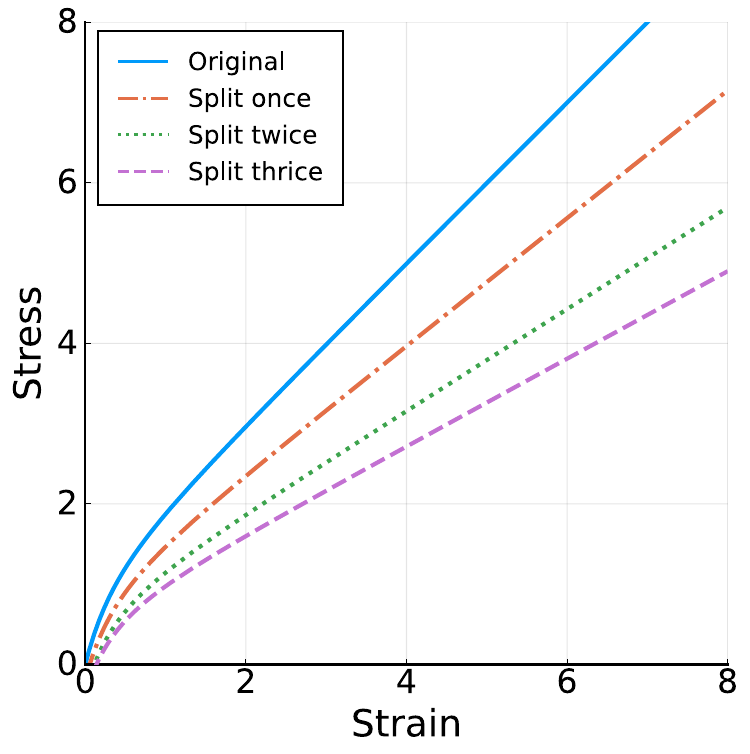}
\end{minipage}
\begin{minipage}[b]{0.01\linewidth}
\end{minipage}
\begin{minipage}[b]{0.49\linewidth}
\centering
\includegraphics[width=\linewidth]{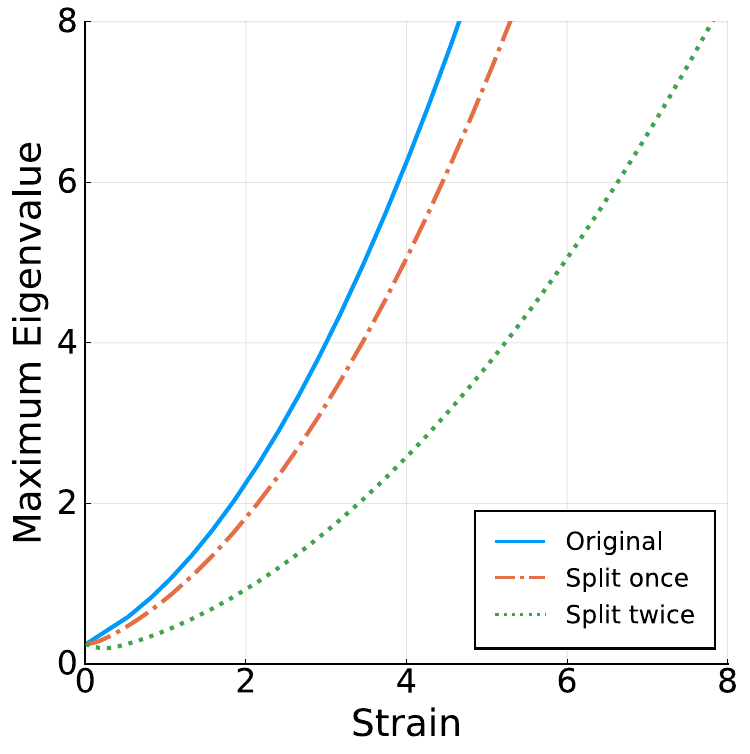}
\end{minipage}

(c) The hexagonal lattice.
\end{center}

\begin{center}
\begin{minipage}[b]{0.49\linewidth}
\centering
\includegraphics[width=\linewidth]{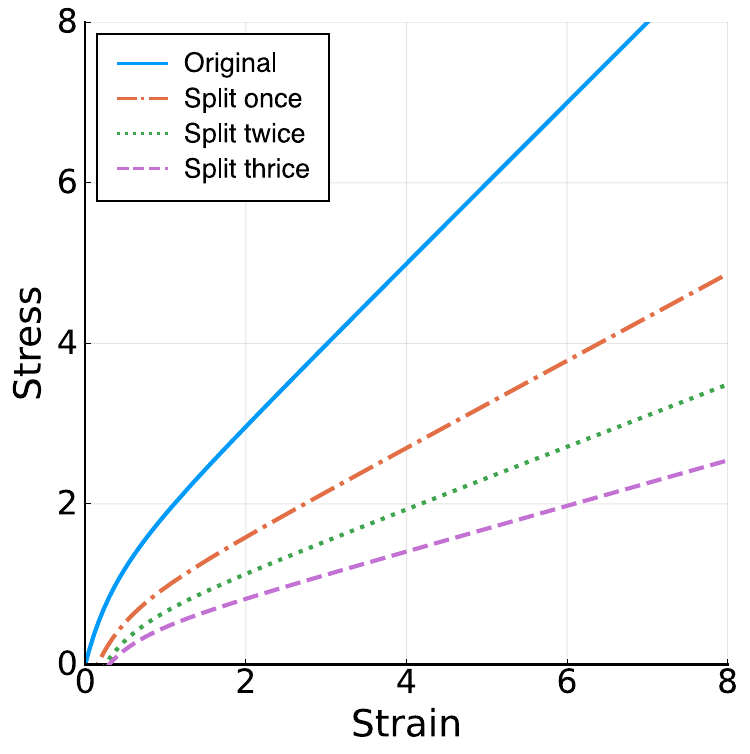}
\end{minipage}
\begin{minipage}[b]{0.01\linewidth}
\end{minipage}
\begin{minipage}[b]{0.49\linewidth}
\centering
\includegraphics[width=\linewidth]{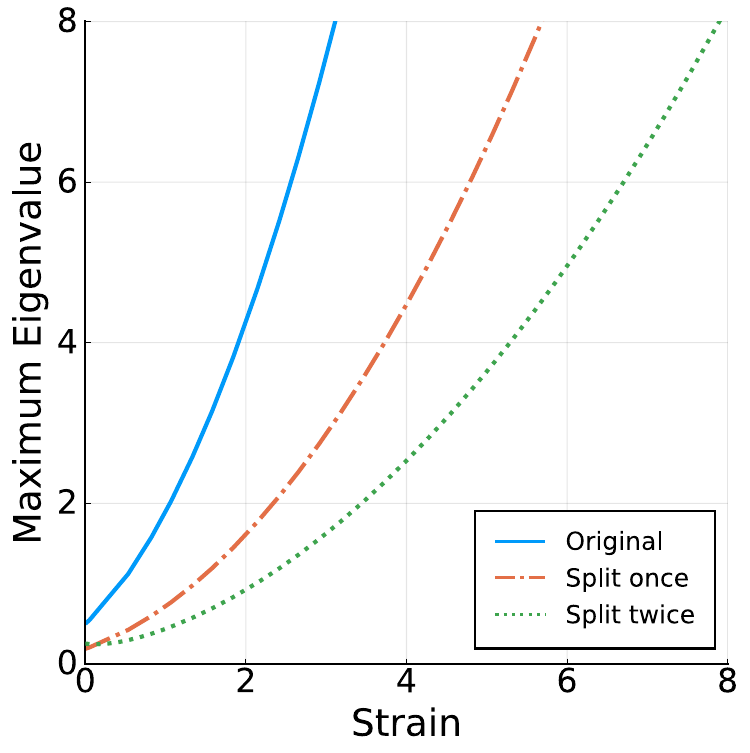}
\end{minipage}

(d) The kagome lattice.
\end{center}

\caption{Each shows the stress--strain curve (left side) and the maximum eigenvalue (right side) of each net shown in Figure~\ref{fig:2d_nets}. \label{fig:2d_graphs}}
\end{figure}

\begin{table}[htbp]
    \caption{Comparison of approximate average permanent strain for each net after split once, twice, and thrice from their initial state.}
    \label{tab:mean_permanent_strain}
    
    \centering
    \begin{tabular}{lcccc}
        \toprule
               & Triangular & Square   & Hexagonal & Kagome \\
        \midrule
        Once   & 0.066846   & 0.047207 & 0.054288  & 0.153974 \\
        Twice  & 0.126739   & 0.116041 & 0.111967  & 0.249558 \\
        Thrice & 0.178081   & 0.139145 & 0.137617  & 0.295855 \\
        \bottomrule
    \end{tabular}
\end{table}

\begin{figure}[tbp]
\begin{minipage}[b]{0.49\linewidth}
\centering
\includegraphics[width=\linewidth]{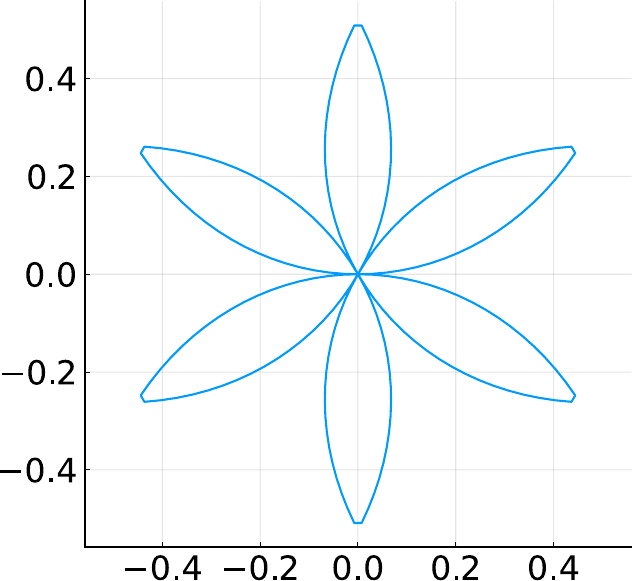}\\
(a)
\end{minipage}
\begin{minipage}[b]{0.01\linewidth}
\end{minipage}
\begin{minipage}[b]{0.49\linewidth}
\centering
\includegraphics[width=\linewidth]{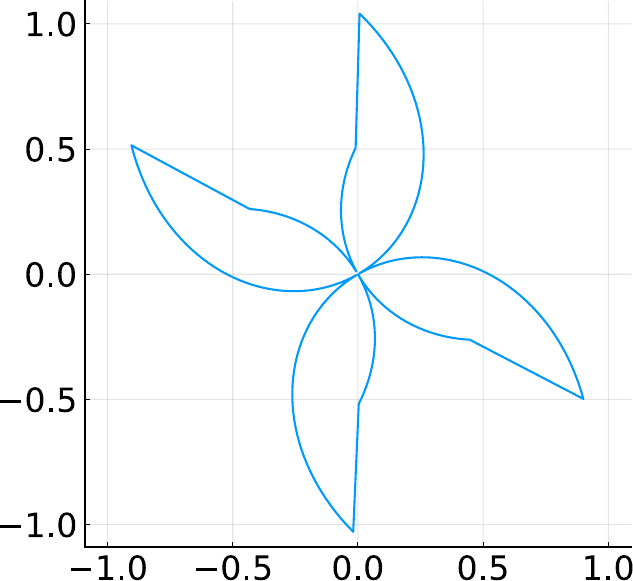}\\
(b)
\end{minipage}

\caption{Representation surfaces of (local) deviation angles of nets. (a) A representation surface of deviation angles of triangular, hexagonal, and kagome lattices illustrated in Figures~\ref{fig:2d_nets}(a), (c), and (d). (b) A representation surface of local deviation angles of the kagome lattice at a vertex. Note that the surface is actually discontinuous in directions where a line segment appears, e.g., the angle $\pi/2$.}
\label{fig:deviation_angle_2d_nets}
\end{figure}

\begin{table}[tbp]
    \caption{Comparison of permanent strains of extending along directions that give the maximal or minimal deviation angles.}
    \label{tab:permanent_strain_vs_deviation_angle}
    
    \centering
    \begin{tabular}{cccccc}
        \toprule
        \multicolumn{2}{c}{Triangular}  & \multicolumn{2}{c}{Hexagonal} & \multicolumn{2}{c}{Kagome} \\
        \small{Minimal} & \small{Maximal} & \small{Minimal} & \small{Maximal} & \small{Minimal} & \small{Maximal} \\
        \midrule
        0.080607 & 0.040676 & 0.060866 & 0.030931 & 0.090976 & 0.189163 \\
        \bottomrule
    \end{tabular}
\end{table}

Finally, we consider the deviation angles of these nets.
In particular, it is interesting to compare the triangular, hexagonal, and kagome nets because, by suitable rotation, they have the same deviation angle for all non-zero vectors (see Figure~\ref{fig:deviation_angle_2d_nets}).
They have the maximal deviation angle $\pi/6$ in six equally spaced directions and the minimal deviation angle $0$ in the other six equally spaced directions.
Table~\ref{tab:permanent_strain_vs_deviation_angle} shows the permanent strains extended in each direction that give the maximal or minimal deviation angles.
(To be precise, the extension directions are shifted by $0.1$ degrees from the angles to not split the edges perpendicular to the extension direction.)

For both triangular and hexagonal lattices, the permanent strain in the direction to give the maximal deviation angle is smaller than that in the direction to give the minimal deviation angle.
By Proposition~\ref{prop:deviation}, if a net contains only one vertex per period, the permanent strain is close to zero by extending in the direction to give the deviation angle close to $\mathrm{arccos}(1/\sqrt{N}) = \pi/4$ (when $N = 2$).
These two nets have four vertices per period, and the permanent strains of the lattices, however, are consistent with Proposition~\ref{prop:deviation}.
This is because the `local' deviation angle at each vertex is the same as the global one.
We define the \emph{local deviation angle} at a vertex as the deviation angle by using its adjacent edges instead of the edges in the net per period.
By Figure~\ref{fig:2d_nets}, the local deviation angle at each vertex of the triangular and hexagonal lattices is the same as the deviation angle of the nets for any direction.
Therefore, the global and local anisotropy coincide, and the nets have the same property as a net containing only one vertex per period.

In contrast, the kagome lattice has different properties from the above two lattices.
As shown in Table~\ref{tab:permanent_strain_vs_deviation_angle}, the permanent strain by extending in the direction to give the minimal deviation angle is smaller than that by extending in the direction to give the maximal one.
The local deviation angle at a vertex is different from the global one, as shown in Figure~\ref{fig:deviation_angle_2d_nets}.
By the figure, we can see that there is a direction that simultaneously gives the maximal and minimal angles in the local and global cases, respectively.
Hence, even if we extend a net in a direction that gives the minimal deviation angle, the angle between the direction and the edge obtained by the split will be large.
So, we obtain a small permanent strain.
This fact is responsible for the difference from the permanent strains of the other two lattices.

From the above, it is clear that deviation angles are effective in quantifying the anisotropy of a net, as seen in Section~\ref{sec:deviation}. 
Furthermore, local tension tensors help predict a permanent strain by extending a net.

\subsection{Another simulation of the kagome lattice}

As discussed in the previous subsection, the kagome lattice has a unique property among the $2$-dimensional nets.
In this subsection, we perform another simulation of the kagome lattice, in which we see that the permanent strain is zero when extended in a specific direction.

In this simulation, the permanent strain of the kagome lattice is \emph{zero} when extended in the direction of $\pi/6$ (see Figure~\ref{fig:kagome_prop}(a)).
This simulation uses the same model as the previous subsection but assumes that the vertices of degree $3$ do not split.
As shown in Figure~\ref{fig:kagome_prop}(a), the first vertex splitting makes a big difference in properties as in the previous subsection. 
However, the second split brings the permanent strain close to zero.
After the third split, the permanent strain is zero.
Since the degree of each vertex is $3$, the net does not split any further, and the permanent strain of the kagome lattice becomes zero (see Table~\ref{tab:kagome_permanent_strain}).
Note that the average permanent strain also behaves similarly to that of extension in the specific direction, as shown in Figure~\ref{fig:kagome_prop}(b) and Table~\ref{tab:kagome_permanent_strain}.

\begin{figure}[htbp]
\begin{center}
\begin{minipage}{0.49\linewidth}
\centering
\includegraphics[width=\linewidth]{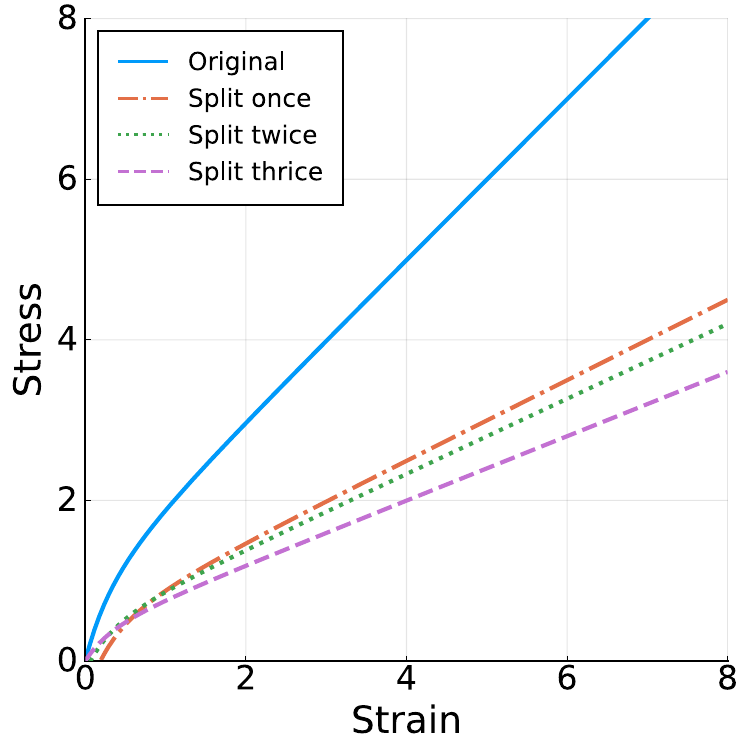}
\end{minipage}
\begin{minipage}{0.01\linewidth}
\end{minipage}
\begin{minipage}{0.49\linewidth}
\centering
\includegraphics[width=\linewidth]{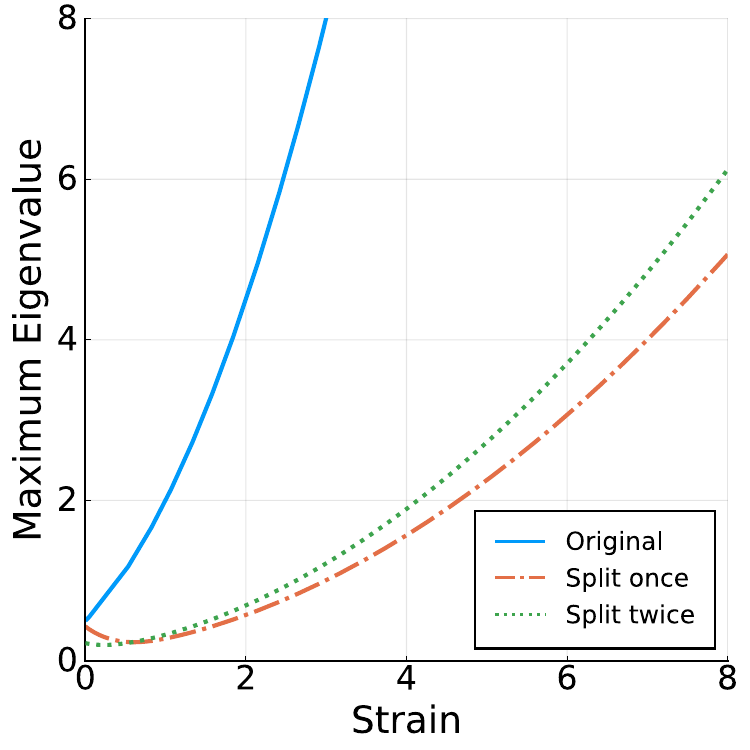}
\end{minipage}

(a) Another simulation of the kagome lattice when extended in the direction of $\pi/6$
\end{center}

\begin{center}
\begin{minipage}{0.49\linewidth}
\centering
\includegraphics[width=\linewidth]{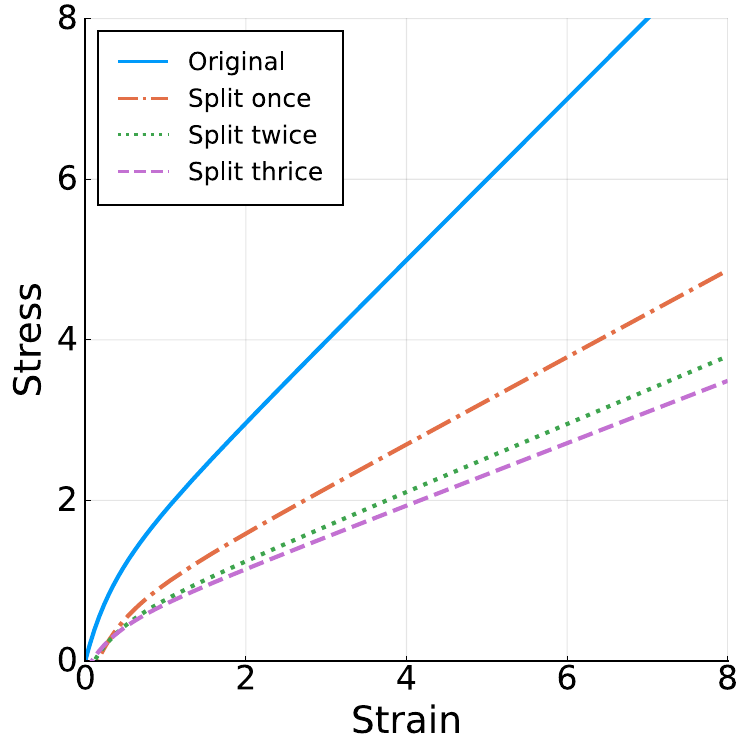}
\end{minipage}
\begin{minipage}{0.01\linewidth}
\end{minipage}
\begin{minipage}{0.49\linewidth}
\centering
\includegraphics[width=\linewidth]{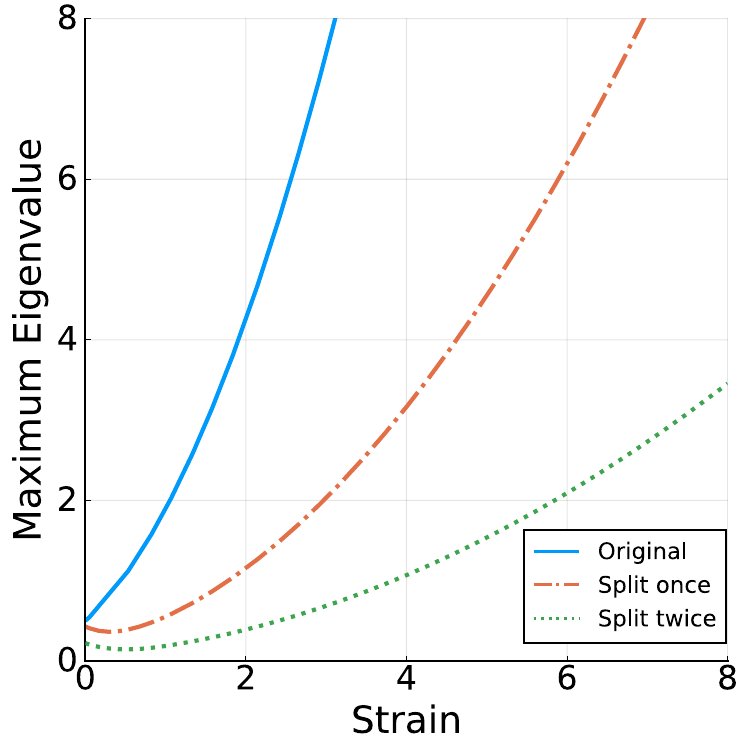}
\end{minipage}

(b) Another simulation of the kagome lattice. They show the average of the curves obtained by stretching in twelve equally spaced directions based on $0$ degrees ($\pi/24$ to be exact).

\end{center}
\caption{Stress--strain curves and maximum eigenvalues obtained from another simulation of the kagome lattice.\label{fig:kagome_prop}}
\end{figure}

\begin{table}[htbp]
    \caption{Approximate average permanent strain and the permanent strain when stretched in the $\pi/6$ direction, obtained from another simulation of the kagome lattice.}
    \label{tab:kagome_permanent_strain}

    \centering
    \begin{tabular}{lcc}
        \toprule
               & Average & $\pi/6$ \\
        \midrule
        Once   & 0.153974   & 0.189207 \\
        Twice  & 0.111967   & 0.064844 \\
        Thrice & 0.081703   & 0 \\
        \bottomrule
    \end{tabular}
\end{table}

Figure~\ref{fig:splitting_kagome} shows the splittings of the kagome lattice in this simulation.
The first splitting occurs at $p_1$, and the second and third splittings occur at $p_2$ and $p_3$.
Note that, as can be seen from Figure~\ref{fig:kagome_prop}(a), the maximum eigenvalue of the net split twice reaches the threshold faster than that of the net split once.
So, the splits at $p_2$ and $p_3$ occur simultaneously.
The result net is a standard realization.
Therefore, the permanent strain of the kagome lattice is zero.

\begin{figure}[htbp]
\begin{minipage}{0.32\linewidth}
\centering
\includegraphics[width=\linewidth]{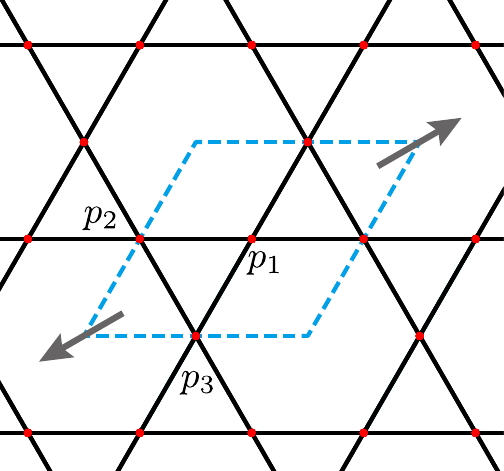}\\
(a)
\end{minipage}
\begin{minipage}{0.32\linewidth}
\centering
\includegraphics[width=\linewidth]{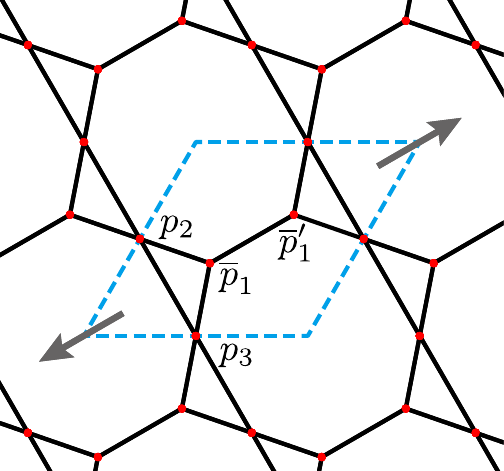}\\
(b)
\end{minipage}
\begin{minipage}{0.32\linewidth}
\centering
\includegraphics[width=\linewidth]{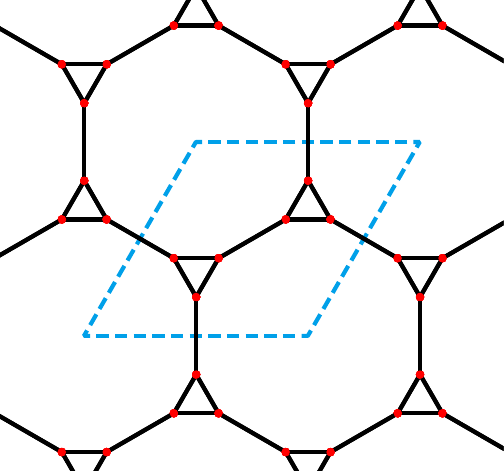}\\
(c)
\end{minipage}
\caption{(a) The kagome lattice.
The arrow indicates the direction of the extension.
(b) The net obtained from the kagome lattice by the splitting of the vertex $p_1$.
(c) The result of the extension of the kagome lattice.\label{fig:splitting_kagome}}
\end{figure}

The result for the kagome lattice in this simulation contrasts with Theorem~\ref{thm:average}.
It suggests that the general version of that does not hold.
However, as discussed in the previous subsection, the simulation for the kagome lattice based on a different model is not a counterexample to the general version of the theorem.
This discussion tells us that the choice of the models is essential.

\section{Biaxial extension}
\label{sec:biaxial}

In this section, we discuss the biaxial extension of 3-dimensional nets.

\subsection{General biaxial extension}

In this subsection, we give a model of general biaxial extension. 
Consider a 3-dimensional harmonic net $(X, \Phi)$, 
which is not necessarily standard. 
We write the tension tensor $(\tau_{ij})_{1 \leq i,j \leq 3} = \calT (X, \Phi)$. 
For $\lambda_{1}, \lambda_{2} > 0$, 
the diagonal matrix 
$A(\lambda_{1}, \lambda_{2}) = \diag(\lambda_{1}, \lambda_{2}, \lambda_{1}^{-1} \lambda_{2}^{-1}) \in SL(3, \bbR)$ 
induces a biaxial extension of the net $(X, \Phi)$. 
The volume $\calV$ per period is constant under deformation. 
Consider the tension tensor 
$\calT (\lambda_{1}, \lambda_{2}) = \calT (A(\lambda_{1}, \lambda_{2}) (X, \Phi)) 
= A(\lambda_{1}, \lambda_{2}) \calT (X, \Phi) A(\lambda_{1}, \lambda_{2})$. 
A stress tensor in the volume-preserving setting is given by 
$(\sigma_{ij})_{1 \leq i,j \leq 3} = (2/\calV) \calT (\lambda_{1}, \lambda_{2}) - cI$ for some $c$. 
Since there is no external force along the direction of the third coordinate, 
we suppose that $\sigma_{33} = 0$. 
Then we obtain the \emph{engineering stresses} under this biaxial extension for $i =1,2$ by 
\[
\sigma^{\mathrm{eng}}_{i} = \sigma_{ii} / \lambda_{i} 
= (2/\calV)(\tau_{ii}\lambda_{i} - \tau_{33}\lambda_{1}^{-2} \lambda_{2}^{-2}\lambda_{i}^{-1}). 
\]
Let $\calE (\lambda_{1}, \lambda_{2}) = \calE (A(\lambda_{1}, \lambda_{2}) (X, \Phi))$. 
Since $\calE (\lambda_{1}, \lambda_{2}) = \tr \calT (\lambda_{1}, \lambda_{2}) 
= \tau_{11}\lambda_{1}^{2} + \tau_{22}\lambda_{2}^{2} + \tau_{33}\lambda_{1}^{-2} \lambda_{2}^{-2}$,
we have the following formula. 

\begin{prop}
$\sigma^{\mathrm{eng}}_{i} = \dfrac{1}{\calV} \dfrac{\partial \calE}{\partial \lambda_{i}}$. 
\end{prop}

\subsection{Uniform biaxial extension}

In this subsection, we show that the uniform biaxial extension can be modeled 
by the inverse operation of the uniaxial extension. 
For $\lambda > 0$, 
let $A(\lambda) = \diag(\lambda, \lambda, \lambda^{-2}) \in SL(3, \bbR)$. 
Then the energy after the extension by $A(\lambda)$ is 
$\calE (\lambda) = (\tau_{11}+\tau_{22})\lambda^{2} + \tau_{33}\lambda^{-4}$. 
The average of the tension tensors in the $(x_{1}, x_{2})$-plane 
is given by $\diag((\tau_{11}+\tau_{22})/2, (\tau_{11}+\tau_{22})/2, \tau_{33})$. 
Then we obtain the engineering stress by 
\[
\sigma_{\mathrm{eng}} 
= \dfrac{2}{\lambda \calV} \left( \dfrac{\tau_{11}+\tau_{22}}{2}\lambda^{2} - \tau_{33}\lambda^{-4} \right) 
= \dfrac{1}{2\calV} \dfrac{d \calE}{d \lambda}. 
\]
The permanent strain $\epsilon_{0}$ is defined by $\sigma_{\mathrm{eng}} (1+\epsilon_{0}) = 0$. 
Hence 
\[
\epsilon_{0} = (2\tau_{33}/(\tau_{11}+\tau_{22}))^{1/6} -1. 
\]

Let us compare it to the uniaxial extension in the direction of $x_{3}$. 
In this case, 
the transformation matrix, the engineering stress, and the permanent strain 
are given as follows: 
\begin{align*}
\tilde{A}(\lambda) & = \diag(\lambda^{-1/2}, \lambda^{-1/2}, \lambda) \in SL(3, \bbR), \\
\tilde{\sigma}_{\mathrm{eng}} & = (2/\calV) (\tau_{33}\lambda - (\tau_{11}+\tau_{22})\lambda^{-2}/2), \\
\tilde{\epsilon}_{0} & = (2\tau_{33}/(\tau_{11}+\tau_{22}))^{-1/3} -1. 
\end{align*}
Then 
$A(\lambda) = \tilde{A}(\lambda^{-2})$, 
$\sigma_{\mathrm{eng}}(\lambda) = -\lambda^{-3} \tilde{\sigma}_{\mathrm{eng}}(\lambda^{-2})$, 
and $1 + \epsilon_{0} = (1 + \tilde{\epsilon}_{0})^{-1/2}$. 
Thus we may regard the uniform biaxial extension as the inverse operation of uniaxial extension.

\subsection{The hyper-kagome lattice}

The hyper-kagome lattice is known as a subnet of the pyrochlore lattice (Figure~\ref{fig:hyper-kagome}(b)).
See, for example, \cite{Chen, Awaga2017} for a discussion on the hyper-kagome lattice.
Note that the hyper-kagome lattice is not a harmonic realization.
So this is not a standard realization either.
We consider a standard net with the same combinatorial structure as the hyper-kagome lattice.
The net is a dual net of the $K_4$ lattice.
In this paper, we will indicate this standard net when we refer to the hyper-kagome lattice.

Like the kagome lattice, the hyper-kagome lattice has zero permanent strain due to extension in a specific (biaxial) direction.
For the primitive lattice shown in Figure~\ref{fig:hyper-kagome}(a), we consider a uniform biaxial extension in a plane parallel to one face of the cube.
First, splittings occur at four vertices that are dual to the edges of the $K_4$ lattice parallel to the plane.
We can obtain the split net by dividing the vertices by the plane perpendicular to the corresponding edges.
Furthermore, splittings occur at the remaining eight vertices if we continue stretching.
These splittings occur in the same way as the previous four vertices.
By splitting, we have a net consisting of shrunk triangles of the hyper-kagome's ones and edges along the $K_4$ lattice connecting them.
The net is called the \emph{srs-a} lattice~\cite{RCSR}, a standard realization for the same translationally symmetry of the hyper-kagome lattice.
Therefore, the permanent strains of these nets are zero (see Table~\ref{tab:permanent_strain_hyper-kagome}).

As we have seen, the hyper-kagome lattice has zero permanent strain by a specific uniform biaxial extension. Unfortunately, however, we have not yet found a direction in which uniaxial extension results in zero permanent strain. Hence, more generally, we have the following natural question.
\begin{ques}
Is there a pair of a net and a direction where the permanent strain becomes zero by an uniaxial extension?
\end{ques}
We hope such a pair exists.

\begin{figure}[htbp]
\begin{minipage}{0.48\hsize}
\centering
\includegraphics[width=0.9\textwidth]{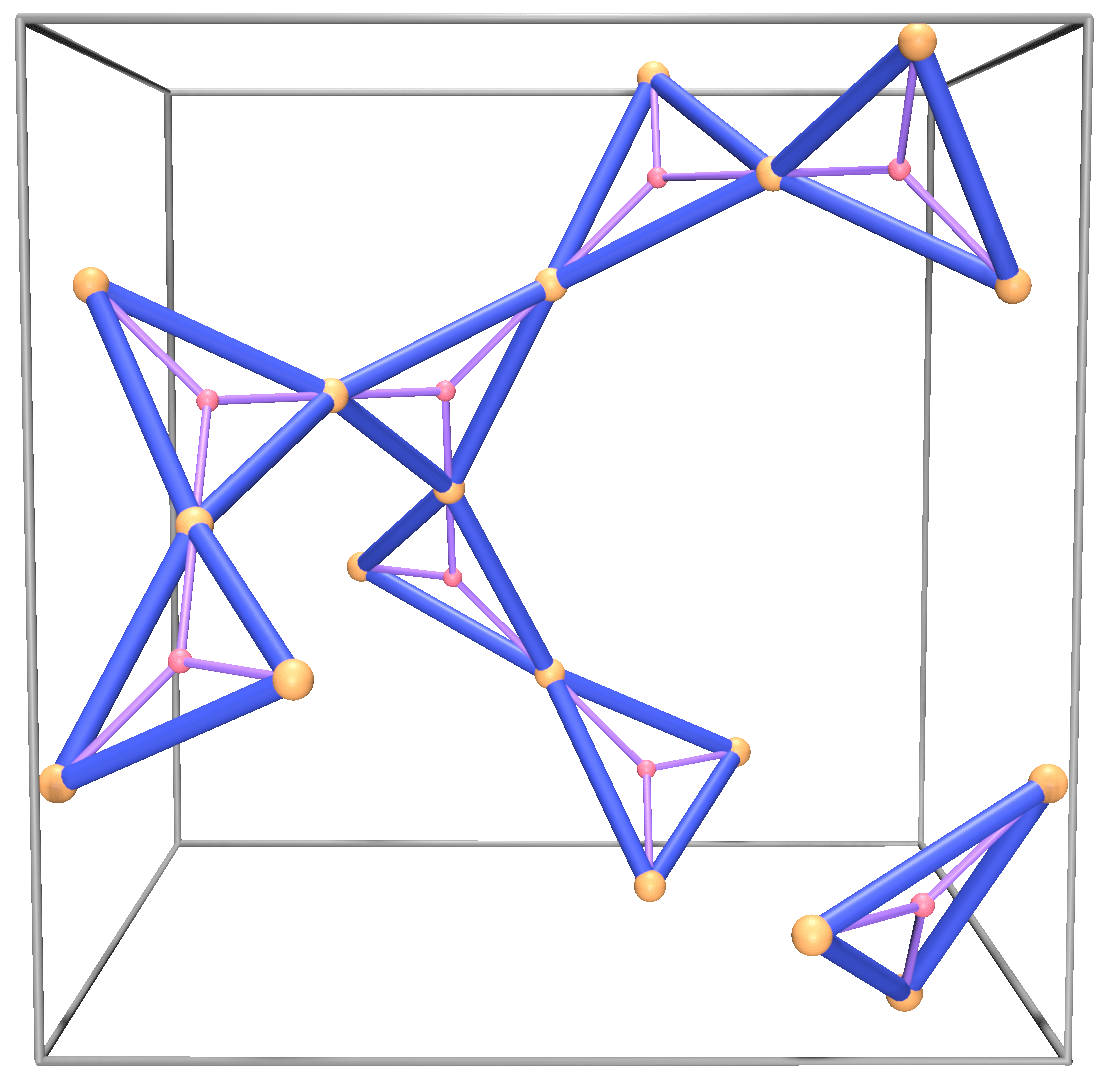}\\
(a)
\end{minipage}
\hfill
\begin{minipage}{0.48\hsize}
\centering
\includegraphics[width=0.9\textwidth]{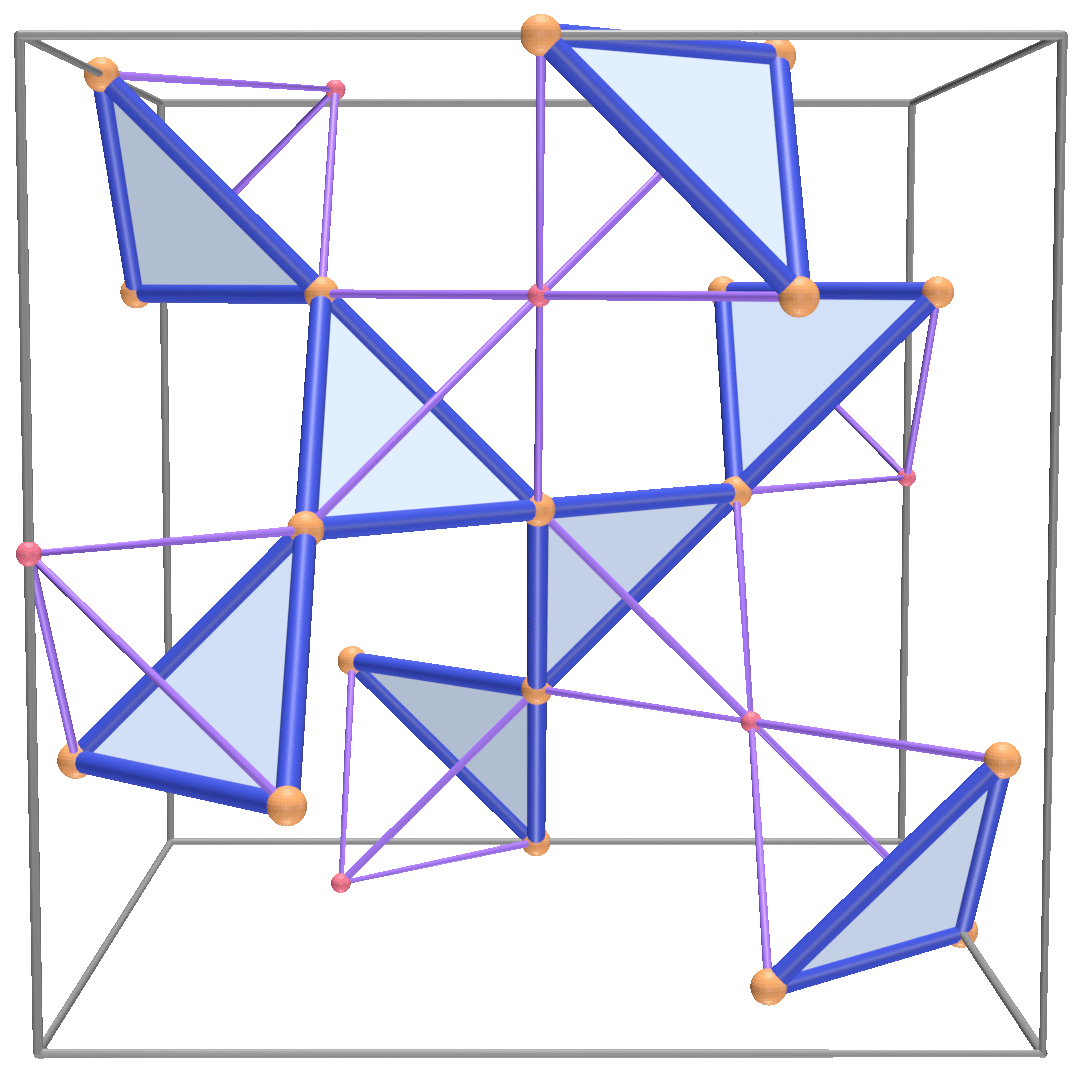}\\
(b)
\end{minipage}

\caption{(a) The modified hyper-kagome lattice with the $K_4$ lattice~\cite{Sunada12}. (b) The hyper-kagome lattice with the pyrochlore lattice (see~\cite{Chen}).}
\label{fig:hyper-kagome}
\end{figure}

\begin{figure}[htbp]
\centering
\includegraphics[width=0.45\linewidth]{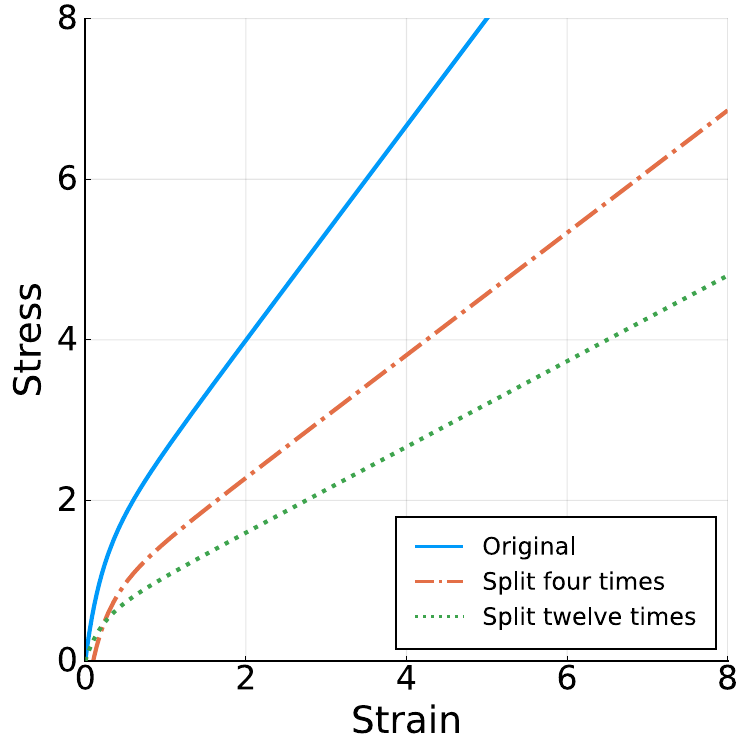}\\
\caption{The stress--strain curve of the hyper-kagome of a uniform biaxial extension orthogonal to $z$-axis. \label{fig:hyperkagome-prop}}
\end{figure}

\begin{table}[htbp]
    \caption{The permanent strain of the hyper-kagome lattice with uniform biaxial extension orthogonal to $z$-axis.}
    \label{tab:permanent_strain_hyper-kagome}
    
    \centering
    \begin{tabular}{lc}
        \toprule
               & Hyper-kagome with uniform biaxial extension \\
        \midrule
        Four times  & 0.097757   \\
        Twelve times  & 0   \\
        \bottomrule
    \end{tabular}
\end{table}

\ifRSPA
\ack{
\else
\section*{Acknowledgements}
\fi
The authors would like to express their gratitude to the anonymous referees for their precious comments.
This research is supported by JST CREST Grant Number JPMJCR17J4 and JSPS KAKENHI Grant Numbers JP21H00978, JP23K17652, and JP23K20791. 
The first author is also supported by the World Premier International Research Center Initiative (WPI), MEXT, Japan, and JSPS KAKENHI Grant Number JP19K14530.
\ifRSPA
}
\fi

\ifRSPA
\bibliographystyle{RS}
\else
\bibliographystyle{plain}
\fi
\bibliography{crest}

\begin{thebibliography}{10}

\bibitem{Aoyagi}
Takeshi Aoyagi, Takashi Honda, and Masao Doi.
\newblock Microstructural study of mechanical properties of the aba triblock
  copolymer using self-consistent field and molecular dynamics.
\newblock {\em J. Chem. Phys.}, 117(17):8153--8161, 2002.

\bibitem{Chen}
Gang Chen and Leon Balents.
\newblock Spin-orbit effects in ${\text{na}}_{4}{\text{ir}}_{3}{\text{o}}_{8}$:
  A hyper-kagome lattice antiferromagnet.
\newblock {\em Phys. Rev. B}, 78:094403, Sep 2008.

\bibitem{Fredrickson}
Glenn Fredrickson.
\newblock {\em {The Equilibrium Theory of Inhomogeneous Polymers}}.
\newblock Oxford University Press, 12 2005.

\bibitem{HINM}
Ryohei Hosoya, Makiko Ito, Ken Nakajima, and Hiroshi Morita.
\newblock Coarse-grained molecular dynamics study of
  styrene-block-isoprene-block-styrene thermoplastic elastomer blends.
\newblock {\em ACS Applied Polymer Materials}, 4(4):2401--2413, 2022.

\bibitem{KMK}
Hiroki Kodama, Hiroshi Morita, and Motoko Kotani.
\newblock A mathematical model of thermoplastic elastomers for analysing the
  topology of microstructures and mechanical properties during elongation.
\newblock {\em Proceedings of the Royal Society A}, 480(2286):20230389, 2024.

\bibitem{KY}
Hiroki Kodama and Ken'ichi Yoshida.
\newblock A mathematical model of network elastoplasticity.
\newblock {\em Proceedings of the Royal Society A: Mathematical, Physical and
  Engineering Sciences}, 478(2260):20210828, 2022.

\bibitem{KS01}
Motoko Kotani and Toshikazu Sunada.
\newblock Standard realizations of crystal lattices via harmonic maps.
\newblock {\em Trans. Amer. Math. Soc.}, 353(1):1--20, 2001.

\bibitem{Nakajima}
Haonan Liu, Xiaobin Liang, and Ken Nakajima.
\newblock Nanoscale strain–stress mapping for a thermoplastic elastomer
  revealed using a combination of in situ atomic force microscopy nanomechanics
  and delaunay triangulation.
\newblock {\em Journal of Polymer Science}, 60(22):3134--3140, 2022.

\bibitem{Awaga2017}
Asato Mizuno, Yoshiaki Shuku, Michio~M. Matsushita, Masahisa Tsuchiizu, Yuuki
  Hara, Nobuo Wada, Yasuhiro Shimizu, and Kunio Awaga.
\newblock 3d spin-liquid state in an organic hyperkagome lattice of mott
  dimers.
\newblock {\em Phys. Rev. Lett.}, 119:057201, Jul 2017.

\bibitem{Polymer}
Hiroshi Morita, Ayano Miyamoto, and Motoko Kotani.
\newblock Recoverably and destructively deformed domain structures in
  elongation process of thermoplastic elastomer analyzed by graph theory.
\newblock {\em Polymer}, 188:122098, 2020.

\bibitem{nye1985physical}
John~Frederick Nye.
\newblock {\em Physical properties of crystals: their representation by tensors
  and matrices}.
\newblock Oxford university press, 1985.

\bibitem{Ogden}
Raymond~William Ogden.
\newblock {\em Non-linear elastic deformations}.
\newblock John Wiley \& Sons, 1984.

\bibitem{RCSR}
Michael O'Keeffe, Maxim~A Peskov, Stuart~J Ramsden, and Omar~M Yaghi.
\newblock The reticular chemistry structure resource (rcsr) database of, and
  symbols for, crystal nets.
\newblock {\em Accounts of chemical research}, 41(12):1782--1789, 2008.

\bibitem{Sunada12}
Toshikazu Sunada.
\newblock {\em Topological crystallography: with a view towards discrete
  geometric analysis}, volume~6 of {\em Surveys and Tutorials in the Applied
  Mathematical Sciences}.
\newblock Springer, 2012.

\bibitem{Treloar}
Leslie R.~G. Treloar.
\newblock {\em The physics of rubber elasticity}.
\newblock Oxford University Press, 1975.

\bibitem{WS}
E.~Wigner and F.~Seitz.
\newblock On the constitution of metallic sodium.
\newblock {\em Phys. Rev.}, 43:804--810, May 1933.

\bibitem{XU2016}
Shanqing Xu, Jianhu Shen, Shiwei Zhou, Xiaodong Huang, and Yi~Min Xie.
\newblock Design of lattice structures with controlled anisotropy.
\newblock {\em Materials \& Design}, 93:443--447, 2016.

\end{thebibliography}

\end{document}